\documentclass[]{article}
\usepackage{fullpage}

\usepackage[pdftex]{graphicx}
\usepackage[cmex10]{amsmath}
\usepackage{algpseudocode}
\usepackage[caption=false,font=footnotesize]{subfig}
\usepackage{amssymb}
\usepackage{amsthm}
\usepackage{thmtools}

\theoremstyle{definition}
\newtheorem{definition}{Definition}
\theoremstyle{plain}
\newtheorem{theorem}[definition]{Theorem}
\newtheorem{lemma}[definition]{Lemma}

\begin{document}

\begin{titlepage}
\title{HSkip+: A Self-Stabilizing Overlay Network for Nodes with Heterogeneous Bandwidths}

\date{}
\author{Matthias Feldotto\\
Heinz Nixdorf Institute \&\\
Department of Computer Science\\
University of Paderborn, Germany\\
Email: feldi@mail.upb.de
\and
Christian Scheideler\\
Theory of Distributed Systems Group\\
Department of Computer Science \\
University of Paderborn, Germany\\
Email: scheideler@upb.de
\and
Kalman Graffi\\
Technology of Social Networks Group\\
Department of Computer Science \\
University of D\"usseldorf, Germany\\
Email: graffi@cs.uni-duesseldorf.de 
\thanks{This work was partially supported by the German Research Foundation (DFG) within the Collaborative Research Centre ``On-The-Fly Computing'' (SFB 901).\newline\newline
This is a long version of a paper published by IEEE in the Proceedings of the 14-th IEEE International Conference on Peer-to-Peer Computing, available at http://dx.doi.org/10.1109/P2P.2014.6934300.\newline\newline
\copyright 2014 IEEE. Personal use of this material is permitted. Permission from IEEE must be obtained for all other uses, in any current or future media, including reprinting/republishing this material for advertising or promotional purposes, creating new collective works, for resale or redistribution to servers or lists, or reuse of any copyrighted component of this work in other works.}
}

\maketitle

\begin{abstract}
In this paper we present and analyze HSkip+, a self-stabilizing overlay
network for nodes with arbitrary heterogeneous bandwidths. HSkip+ has the same
topology as the Skip+ graph proposed by Jacob et al.
\cite{Jacob:2009:DPT:1582716.1582741} but its self-stabilization mechanism
significantly outperforms the self-stabilization mechanism proposed for Skip+.
Also, the nodes are now ordered according to their bandwidths and not
according to their identifiers. Various other solutions have already been
proposed for overlay networks with heterogeneous bandwidths, but they are not
self-stabilizing. In addition to HSkip+ being self-stabilizing, its
performance is on par with the best previous bounds on the time and work for
joining or leaving a network of peers of logarithmic diameter and degree and
arbitrary bandwidths. Also, the dilation and congestion for routing messages
is on par with the best previous bounds for such networks, so that HSkip+
combines the advantages of both worlds. Our theoretical investigations are
backed by simulations demonstrating that HSkip+ is indeed performing much
better than Skip+ and working correctly under high churn rates.
\end{abstract}
\end{titlepage}

\section{Introduction}
\label{sec:Introduction}

Peer-to-peer systems have become very popular for a variety of reasons. For
example, the fact that peer-to-peer systems do not need a central server means
that individuals can search for information or cooperate without fees or an
investment in additional high-performance hardware. Also, peer-to-peer systems
permit the sharing of resources (such as computation and storage) that
otherwise sit idle on individual computers. However, the absence of any
trusted anchor like a central server also has its disadvantages since churn
and adversarial behavior has to be managed by the peers without outside help.
One promising approach that has been investigated in recent years is to use
topological self-stabilization, i.e., the overlay network of the peer-to-peer
system can recover its topology from any state, as long as it is initially
weakly connected. Various topologies have already been considered, but so far
mostly the case has been studied that the peers have the same resources
(concerning speed, storage, and bandwidth) whereas in reality the available
resources may differ significantly from peer to peer. An exception is a
self-stabilizing system for peers with non-uniform storage
\cite{Kniesburges2013}, but no self-stabilizing peer-to-peer system for peers
with non-uniform bandwidth has been proposed yet. Due to the current development, especially with mobile devices, this property becomes more important and is often the bottleneck in the system. This paper is the first to
propose a system which considers the bandwidth in its design.

\subsection{Our model}

\subsubsection{Network model}
\label{sec:NetworkModel}

Similar to Nor et al. \cite{Nor:2011:CSD:2050613.2050640}, we assume that we
have a (potentially dynamic) set $V$ of $n$ nodes with unique identifiers.
Each node $v$ maintains a set of variables (determined by the protocol) that
define the {\em state} of node $v$. For each pair of nodes $u$ and $v$ we have
a channel $C_{u,v}$ that holds messages that are currently in transit from $u$
to $v$. We assume that the channel capacity is unbounded, there is no message
loss, and the messages are delivered asynchronously to $v$ in FIFO order. The sequence of all
messages stored in $C_{u,v}$ constitutes the {\em state} of $C_{u,v}$.

Whenever a node $u$ stores a reference of node $v$, we consider that as an
{\em explicit edge} $(u,v)$, and whenever a channel $C_{u,v}$ holds a message
with a reference of node $w$, we consider that as an {\em implicit edge}
$(v,w)$. Node references are assumed to be atomic and read-only, i.e., they
cannot be split, encoded, or altered. They can only be deleted or copied to
produce new references that can be sent to other nodes.
If there is a reference of a node that is not in the system any more, we
assume that this can be detected by the nodes, so that without loss of
generality we can assume that only references of nodes that are still in the
system are present in the nodes and the channels. Whenever a message in
$C_{u,v}$ contains a node reference $w$, it also contains $w$'s bandwidth and
identifier so that this information can be corrected in $v$ if needed (though
it might initially be wrong).
Only point-to-point communication is possible, and the nodes can only send
messages along explicit edges (since they are not yet aware of the endpoint of
an implicit edge). Whenever node $u$ sends a message along an explicit edge
$(u,v)$, it is transferred to $C_{u,v}$. Let $E_e$ denote the set of all
explicit edges and $E_i$ denote the set of all implicit edges. The overlay
network formed by the system is defined as a directed graph $G=(V,E)$ with
$E=E_e \cup E_i$. With $G_e=\left(V,E_e\right)$ respectively
$G_i=\left(V,E_i\right)$ we define the network which only consists of the
explicit respectively implicit edges. The degree of a node $v$ in $G$ is equal
to its degree in $G_e$ (i.e., the number of explicit edges of the form
$(v,w)$), and the diameter of $G$ is equal to the diameter of $G_e$. We define
the specific state of the network at a time $t$ with $G\left(t\right),
E\left(t\right), E_e\left(t\right)$ and $E_i\left(t\right)$.

\subsubsection{Computational model}
\label{sec:ComputationalModel}

A {\em program} is composed of a set of variables and actions. An {\em action}
has the form $\langle label \rangle: \; \langle guard \rangle \rightarrow
\langle command \rangle$. $label$ is a name to differentiate between actions,
$guard$ can be an arbitrary predicate based on the state of the node executing
the action, and $command$ represents a sequence of commands. A $guard$ of the
form $received(m)$ is true whenever a message $m$ has been received (and not
yet processed) by the corresponding node. An action is {\em enabled} if its
guard is true and otherwise {\em disabled}.

The {\em node state} of the system is the combination of all node states, and
the {\em channel state} of the system is the combination of all channel
states. Both states together form the \emph{(program) state} of the system. A
{\em computation} is an infinite fair sequence of states such that for each
state $s_t$, the next state $s_{t+1}$ is obtained by executing the commands of
an action that is enabled in $s_t$. So for simplicity we assume that only one
action can be executed at a time, but our results would also hold for the
distributed scheduler, i.e., only one action can be executed {\em per node} at
a time. We assume two kinds of fairness of computation: weak fairness for
action execution and fair message receipt. {\em Weak fairness} of action
execution means that no action will be enabled without being executed for an
infinite number of states (i.e., no action will starve, and actions that are
enabled for an infinite number of states will be executed infinitely often).
{\em Fair message receipt} means that every message will eventually be
received (and therefore processed due to weak fairness).

\subsubsection{Topological self-stabilization}
\label{sec:Self-Stabilization}

Next we define topological self-stabilization, which goes back to the idea by
Dijkstra \cite{Dijkstra:1974:SSS:361179.361202} and is summarized by Schneider
\cite{Schneider:1993:SEL:151254.151256}. In the topological self-stabilization
problem we start with an arbitrary state (with a finite number of nodes,
channels, and messages) in which $G$ is weakly connected, and a {\em legal
state} is any state where the topology of $G_e$ has the desired form and the
information that the nodes have about their neighbors is correct. We assume
without loss of generality that $G$ is initially weakly connected, because if
not, then we would just focus on any of the weakly connected components of $G$
and would prove topological self-stabilization for that component. In order to
show topological self-stabilization, two properties need to be shown:

\begin{definition}[Convergence] \label{def:Convergence}
For any initial state in which $G$ is weakly connected, the system eventually
reaches a legal state.
\end{definition}

\begin{definition}[Closure] \label{def:Closure}
Whenever the system is in a legal state, then it is guaranteed to stay in a
legal state, provided that no faults or changes in the node properties (in our
case, the bandwidth) happen.
\end{definition}

\subsection{Related work}

Topological self-stabilization has recently attracted a lot of attention.
Various topologies have been considered such as simple line and ring networks
(e.g., \cite{Shaker:2005:SSR:1099548.1100573,Gall:2010:TCD:2128719.2128746}),
skip lists and skip graphs (e.g., \cite{Nor:2011:CSD:2050613.2050640,
Jacob:2009:DPT:1582716.1582741}), expanders
\cite{Dolev:2013:SDS:2451995.2452270}, the Delaunay graph
\cite{Jacob:2012:THT:2364637.2364951}, the hypertree
\cite{Dolev:2004:HSP:1025126.1025933}, and Chord
\cite{Kniesburges:2011:RSC:1989493.1989527}. Also a universal protocol for
topological self-stabilization has been proposed
\cite{Berns:2011:BSO:2050613.2050620}. However, none of these works consider
nodes with heterogeneous bandwidths.

Various network topologies have been suggested to interconnect nodes with
heterogeneous bandwidths. While
\cite{Nejdl:2003:SRC:775152.775229, Srivatsa:2004:SUP:1009385.1010039} do not
provide any formal guarantees and just evaluate their constructions via
experiments, \cite{Bhargava:2004:PDO:1007912.1007938,
Scheideler:2009:DOH:1575973.1576024} give formal guarantees, but (like for the
experimental papers) no self-stabilizing protocol has been proposed for these
networks.
There has also been extensive work on networks with heterogeneous bandwidths
in the context of streaming applications (see, e.g., \cite{MW12} and the
references therein) but the focus is more on coding schemes and optimization
problems, so it does not fit into our context.

Probabilistic approaches like ours have the advantage of better graph
properties (e.g., a logarithmic expansion) compared to deterministic variants
(e.g., \cite{Awerbuch:2004:HLD:982792.982836}).
\subsection{Our Contribution}

We modified the protocol proposed for the self-stabilizing Skip+ graph
\cite{Jacob:2009:DPT:1582716.1582741} to organize the nodes in a more
effective way using the same topology. However, the nodes are not ordered
according to their labels but according to their bandwidths. Due to this we
call our graph the HSkip+ graph. Improvements of our
construction over previous work are:

\begin{itemize}
 
\item We prove self-stabilization under the asynchronous message passing model
whereas in \cite{Jacob:2009:DPT:1582716.1582741} it was only shown for the
synchronous message passing model.

\item In simulations, our Skip+ protocol has basically the same
self-stabilization time as the original Skip+ protocol, but it spends
significantly less work (in terms of messages that are exchanged between the
nodes) in order to reach a legal state. Furthermore, our overlay is working
correctly under a churn rate of nearly 50\%.

\item When a node joins or leaves in a legal state, then the worst case work in
\cite{Jacob:2009:DPT:1582716.1582741} is $O(\log^4 n)$ w.h.p. whereas our
protocol just needs a worst case work of $O(\log^2 n)$ w.h.p., and a
worst-case time of $O(\log n)$ w.h.p. in the synchronous message passing
model, to get back to a legal state. The work and time bounds are on par with
the previously best (non-self-stabilizing) network of logarithmic diameter and
degree for peers with heterogeneous bandwidths
\cite{Scheideler:2009:DOH:1575973.1576024}.

\item Also the competitiveness concerning the congestion of arbitrary
routing problems in HSkip+ is on par with the previously best
(non-self-stabilizing) network of logarithmic diameter and degree for peers
with heterogeneous bandwidths \cite{Bhargava:2004:PDO:1007912.1007938}.
\end{itemize}

Hence, our HSkip+ construction combines the best results of both worlds
(self-stabilizing networks and scalable networks for heterogeneous nodes).

\subsection{Organization of the Paper}

This paper is structured as follows: In Section \ref{sec:Theory} we present
our topology and the associated self-stabilizing algorithm. We show its
convergence and closure. Furthermore, we look at the handling of external
dynamics and routing in our network. Section \ref{sec:Simulations} presents
our simulation results, especially the comparison of Skip+ and HSkip+.
Finally, we end the paper in Section \ref{sec:Conclusion} with a conclusion.

\section{Theoretical Analysis}
\label{sec:Theory}

\subsection{HSkip+ Topology}
\label{sec:Topology}

We now present the desired topology for our problem which we call
\emph{HSkip+}. It is the same as the \emph{Skip+} topology introduced by Jacob
et al. \cite{Jacob:2009:DPT:1582716.1582741}, which is based on skip
graphs \cite{Aspnes:2003:SG:644108.644170}, but the ordering of the nodes is
different. Instead of using fixed node labels for the ordering, the bandwidth
values are used. Also, new rules are used since they turned out to consume clearly less work than the rules proposed for \emph{Skip+}.

As stated in our network model (cf. Sec. \ref{sec:NetworkModel}), the system
forms a directed graph $G=(V,E)$. Each node $v \in V$ has several internal
variables which define the internal state of the node $v$:

\begin{itemize}

\item $v.id$ is the unique, immutable identifier of node $v$.

\item $v.rs$ is an immutable pseudo-random bit string of node $v$.

\item $v.bw$ is the current bandwidth of node $v$, which is modifiable during
the execution. W.l.o.g., we assume that all bandwidths are
unique (which is easy to achieve given that the node identifiers are unique).

\item $v.nh$ is the neighborhood of node $v$, i.e., the set of all nodes
whose references are stored in $v$.
\end{itemize}

We introduce some auxiliary functions for the internal variables, and
especially for the bit string $v.rs$:
 
\begin{itemize}

\item $\mathrm{prefix}(v, i) = v.rs[0 \ldots i-1]$ is the prefix of length $i$
of $v.rs$.

\item $\mathrm{commonPrefix}(v, w) = \mathrm{argmax}_i \{\mathrm{prefix}(v,i)$
$=\mathrm{prefix}(w,i)\}$ is the length of the maximal common prefix of $v.rs$
and $w.rs$.

\item $\mathrm{level}(v) = \max_{w \in v.nh}{\{\mathrm{commonPrefix}(v,w)\}}$
is called the current level of node $v$ and will be used later for the
grouping of the nodes.

\item $\mathrm{deg}(v) = \left|v.nh\right|$ is called the current degree of
node $v$, the number of neighbors.
\end{itemize}

As mentioned before, we are
aiming at maintaining the HSkip+ graph among the nodes. For the HSkip+ graph
we need a series of definitions.

\begin{definition}[Component of HSkip+] \label{def:Component}
Two nodes $v$ and $w$ belong to the same \emph{component} at level $i$ of
HSkip+ if their bit values $v.rs$ and $w.rs$ share the same prefix of length
$i$, formally $\mathrm{component}(v,i) = \left\{ w \in V \middle|
\mathrm{commonPrefix}(v,w) \geq i \right\}$. The nodes in each component are
ordered according to their bandwidths. A component is called \emph{trivial} if
there is only one node in the component.
\end{definition}

Components exist at each level as long as they are non-trivial, which means
there would be only one node in a component. Therefore, we can define the
level of HSkip+ as the number of levels needed to represent all non-trivial
components:

\begin{definition}[Level of HSkip+] \label{def:Level}
The \emph{level} of the network $G=(V,E)$ is defined by the maximum level $i$
such that two nodes share a common prefix of length $i$. Formally,
$\mathrm{level}(G) = \max_{v,w \in V, v \neq w}{\mathrm{commonPrefix}(v,w)}$.
\end{definition}

In addition to the regular linked list for each component we have further
edges to get a more stable neighborhood and to allow local checking of the
correctness. Each node $v$ is connected at level $i$ to at least one node
$w_0$ and one node $w_1$ which share the same prefix of length $i$ and have
the next bit as $0$ respectively $1$:

\begin{definition}[Farthest Neighbors of HSkip+]
We define the farthest predecessors of node $v$ as
\begin{eqnarray*}
  \mathrm{farthestPred}(v,i,b)
    & = & \mathrm{argmax}_{u \in \left\{u \in V \middle| u.bw > v.bw \right\}} \\
    & & {\left\{\mathrm{prefix}(v,i)\cdot b = \mathrm{prefix}(u,i+1)\right\}}
    \\
  \mathrm{farthestPred}(v,i)
    & = & \mathrm{argmax}_{u \in \left\{u \in V \middle| u.bw > v.bw
    \right\}} \\
    & & \hspace*{-1cm} {\min_{b\in \left\{0, 1\right\}}{\left\{\mathrm{prefix}(v,i)\cdot b =
      \mathrm{prefix}(u,i+1)\right\}}}
\end{eqnarray*}
and the farthest successors as
\begin{eqnarray*}
  \mathrm{farthestSucc}(v,i,b)
  & = & \mathrm{argmin}_{w \in \left\{w \in V \middle|
    w.bw < v.bw \right\}} \\
  & & {\left\{\mathrm{prefix}(v,i)\cdot b = \mathrm{prefix}(w,i+1)\right\}} \\
  \mathrm{farthestSucc}(v,i)
  & = & \mathrm{argmin}_{w \in \left\{w \in V \middle|
    w.bw < v.bw \right\}} \\
  & & \hspace*{-1cm} {\max_{b\in \left\{0, 1\right\}}{\left\{\mathrm{prefix}(v,i)\cdot b =
    \mathrm{prefix}(w,i+1)\right\}}}
\end{eqnarray*}
\end{definition}

Also, all nodes between the farthest predecessor and successor in each level are connected. This property aims at a stable
neighborhood which prepares the next higher level as the linked list of level
$i+1$ is already available. Formally we can define the neighborhood range at
level $i$ with the help of the components:

\begin{definition}[Range of HSkip+] \label{def:Range}
The node $v$ is connected at level $i$ to all nodes $w \in
\mathrm{component}(v,i)$ with $w.bw \leq \mathrm{farthestPred}(v,i).bw$ and
$w.bw \geq \mathrm{farthestSucc}(v,i).bw$, which we call the \emph{range} or
\emph{neighbors} of node $v$ at level $i$.
\end{definition}

Furthermore, we define different neighborhood shortcuts:

\begin{definition}[Neighbors of HSkip+]
Let
\begin{itemize}
\item $\mathrm{preds}(v,i) = \left\{u \in \mathrm{neighbors}(v,i) \middle| u.bw
  > v.bw \right\}$
\item $\mathrm{succs}(v,i) = \left\{w \in \mathrm{neighbors}(v,i) \middle| w.bw
  < v.bw \right\}$
\item $\mathrm{closestPred}(v,i) = \mathrm{argmin}_{u \in \mathrm{preds(v,i)}}
  {u.bw}$
\item $\mathrm{closestSucc}(v,i) = \mathrm{argmax}_{w \in \mathrm{succs(v,i)}}
  {w.bw}$
\end{itemize}
\end{definition}

With the three definitions of components, level and range (cf. Def.
\ref{def:Component}, \ref{def:Level} and \ref{def:Range}) in the HSkip+
topology we can now define the desired network. See Fig.
\ref{fig:ExampleNetwork} for an example.

\begin{definition}[HSkip+]
    \label{def:HSkip+}
    The HSkip+ network is defined by $G^{HSkip+} = \left(V, E^{HSkip+}\right)$ with
    \[E^{HSkip+}=\left\{(v,w) \middle| \exists i \in [0, \mathrm{level}]: w \in \mathrm{range}(v, i) \right\}\]
\end{definition}

\begin{figure}[htb]
    \centering
    \includegraphics[scale=.5]{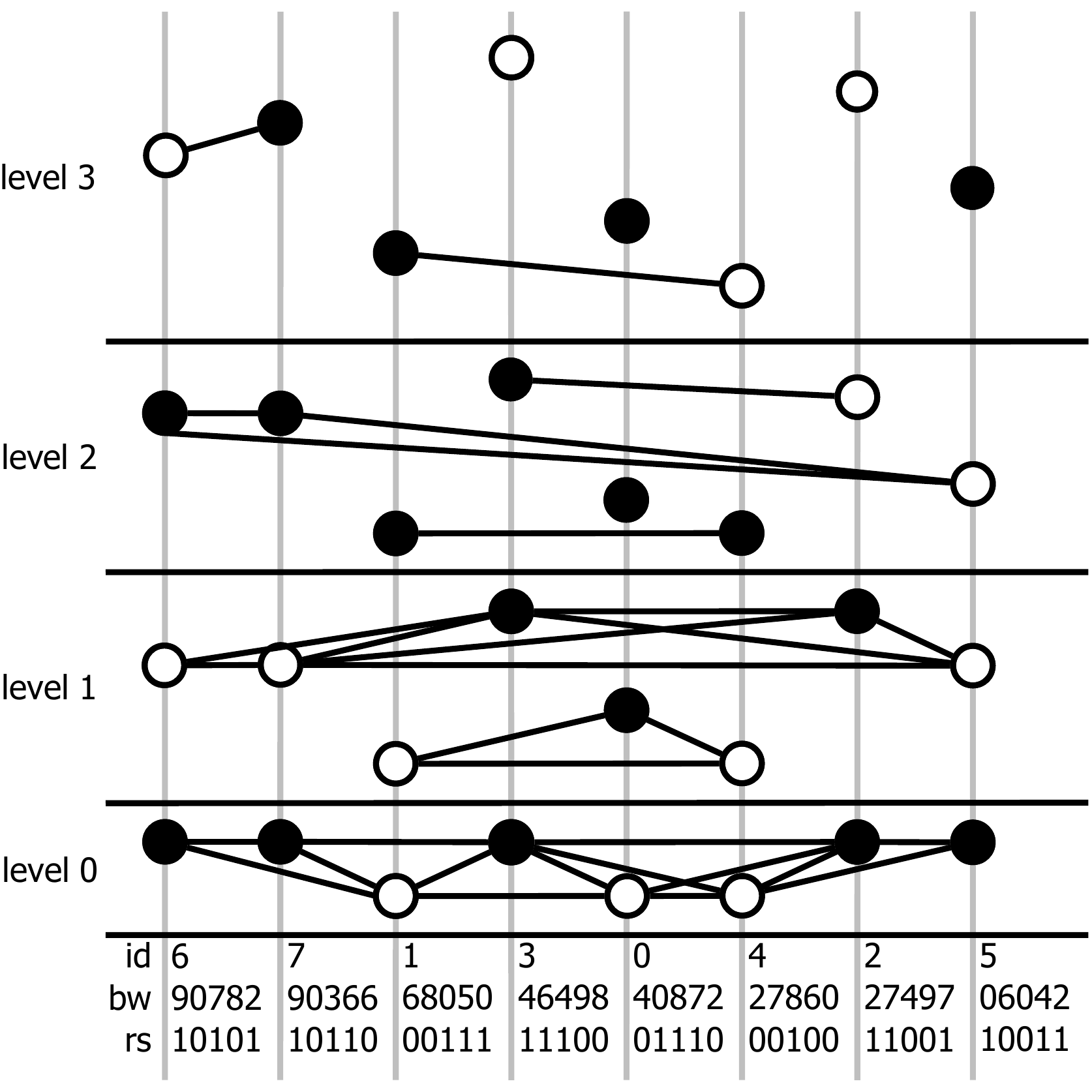}
    \caption[Example for the HSkip+ Network]{Example for the HSkip+ topology with eight nodes.}
    \label{fig:ExampleNetwork}
\end{figure}

Figure~\ref{fig:ExampleNetwork} shows an example network with 8 nodes and 4 levels. The image visualizes the edges caused by the different levels in the HSkip+ topology.

\subsection{HSkip+ Algorithm}
\label{sec:Algorithm}

To reach the presented topology, we now introduce a self-stabilizing algorithm
whose correctness we will prove afterwards. All following operations are
executed at node $v$. The algorithm works in the asynchronous message passing
model presented in Section \ref{sec:Introduction}. We just use two types of
guards: $true$ and $received(m)$. $true$ means that the action is continuously
enabled, which implies that it is executed infinitely often. In addition to the definitions in Sec.~\ref{sec:Topology} we take use of local variants (e. g. $\mathrm{localFarthestPred}(v,i)$) which represent the cached information at node $v$.

\begin{figure}[htb]
    \fbox{\parbox{0.975\columnwidth}{
    \begin{algorithmic}\footnotesize
        \State $true \rightarrow$
        \State CheckNeighborhood()
        \State IntroduceNode()
        \State IntroduceClosestNeighbors()
        \State LinearizeNeighbors()
    \end{algorithmic}
    }}
    \caption[Periodic Actions]{Action called periodically by node $v$.}
    \label{alg:PeriodicActions}
\end{figure}

The \emph{CheckNeighborhood()} function checks if all nodes which are in the
neighborhood of the node $v$ are needed for the topology (cf. Fig.
\ref{alg:CheckNeighborhood}). If a node $w$ is no longer needed in any level,
it is removed from the neighborhood of $v$ and forwarded to another node $x$,
concretely the node in the neighborhood with the longest common prefix,
because it is the most promising node which could include the node in its own
neighborhood at any level.

\begin{figure}[htb]
    \fbox{\parbox{0.975\columnwidth}{
    \begin{algorithmic}\footnotesize
        \Function{CheckNeighborhood}{}
            \ForAll{$w \in v.nh$}
                \If{CheckNode($w$) = $false$}
                    \State $v.nh = v.nh \backslash \{w\}$
                    \State send $m=(\mathrm{\emph{build}}, w)$\\
                    \hspace{\algorithmicindent}\hspace{\algorithmicindent}\hspace{\algorithmicindent}\hspace{\algorithmicindent} to node $\mathrm{argmax}_{x\in v.nh}{\mathrm{commonPrefix}(w,x)}$
                \EndIf
            \EndFor
        \EndFunction
    \end{algorithmic}
    }}
    \caption[CheckNeighborhood()]{\emph{CheckNeighborhood()} inspects the nodes in $v.nh$.}
    \label{alg:CheckNeighborhood}
\end{figure}

The test if a node $w$ is really needed for node $v$ is done by the
\emph{CheckNode()} function (cf. Fig. \ref{alg:CheckNode}). For the node $w$
in each level $i$ it is checked if the node $w$ is in the range (cf. Def.
\ref{def:Range}) of node $v$. If this is the case for at least one level, the
node $w$ is needed in the neighborhood of $v$.


\begin{figure}[htb]
    \fbox{\parbox{0.975\columnwidth}{
    \begin{algorithmic}\footnotesize
        \Function{CheckNode}{node $w$}
            \State $needed \leftarrow false$
            \For{$i=0 \text{ to } \mathrm{level}(v)$}
                \If{$\mathrm{prefix}(v,i) = \mathrm{prefix}(w,i)$\\
                    \hspace{\algorithmicindent}\hspace{\algorithmicindent}\hspace{\algorithmicindent}$\wedge (\mathrm{localFarthestPred}(v,i) = \bot$\\
                    \hspace{\algorithmicindent}\hspace{\algorithmicindent}\hspace{\algorithmicindent}\hspace{\algorithmicindent}$\vee w.bw \leq \mathrm{localFarthestPred}(v,i).bw)$\\
                    \hspace{\algorithmicindent}\hspace{\algorithmicindent}\hspace{\algorithmicindent}$\wedge (\mathrm{localFarthestSucc}(v,i) = \bot$\\
                    \hspace{\algorithmicindent}\hspace{\algorithmicindent}\hspace{\algorithmicindent}\hspace{\algorithmicindent}$\vee w.bw \geq \mathrm{localFarthestSucc}(v,i).bw)$\\\hspace{\algorithmicindent}\hspace{\algorithmicindent}}
                    \State $needed \leftarrow true$
                \EndIf
            \EndFor
            \State\Return $needed$
        \EndFunction
    \end{algorithmic}
    }}
    \caption[CheckNode()]{\emph{CheckNode()} tests if a node $w$ is needed in $v.nh$.}
    \label{alg:CheckNode}
\end{figure}

\begin{figure}[htb]
    \fbox{\parbox{0.975\columnwidth}{
    \begin{algorithmic}\footnotesize
        \Function{IntroduceNode}{}
            \ForAll{$w \in v.nh$}
                \State send $m=(\mathrm{\emph{build}}, v)$ to node $w$
            \EndFor
        \EndFunction
    \end{algorithmic}
    }}
    \caption[IntroduceNode()]{\emph{IntroduceNode()} introduces the node $v$ to all neighbors.}
    \label{alg:IntroduceNode}
\end{figure}

\begin{figure}[htb]
    \fbox{\parbox{0.975\columnwidth}{
    \begin{algorithmic}\footnotesize
        \Function{IntroduceClosestNeighbors}{}
            \For{$i = 0 \text{ to } \mathrm{level}(v)$}
                \If{$\mathrm{localClosestPred}(v,i) \neq \bot$}
                    \ForAll{$w \in \mathrm{localNeighbors}(v,i)$}
                        \State send $m=(\mathrm{\emph{build}}, \mathrm{localClosestPred}(v,i))$\\
                        \hspace{\algorithmicindent}\hspace{\algorithmicindent}\hspace{\algorithmicindent}\hspace{\algorithmicindent}\hspace{\algorithmicindent} to node $w$
                    \EndFor
                \EndIf
                \If{$\mathrm{localClosestSucc}(v,i) \neq \bot$}
                    \ForAll{$w \in \mathrm{localNeighbors}(v,i)$}
                        \State send $m=(\mathrm{\emph{build}}, \mathrm{localClosestSucc}(v,i))$\\
                        \hspace{\algorithmicindent}\hspace{\algorithmicindent}\hspace{\algorithmicindent}\hspace{\algorithmicindent}\hspace{\algorithmicindent} to node $w$
                    \EndFor
                \EndIf
            \EndFor
        \EndFunction
    \end{algorithmic}
    }}
    \caption[IntroduceClosestNeighbors()]{\emph{IntroduceClosestNeighbors()} introduces the closest neighbors.}
    \label{alg:IntroduceClosestNeighbors}
\end{figure}

\begin{figure}[htb]
    \fbox{\parbox{0.975\columnwidth}{
    \begin{algorithmic}\footnotesize
        \Function{LinearizeNeighbors}{}
            \For{$i = 0 \text{ to } \mathrm{level}(v)$}
                \For{$j=0 \text { to } \left|\mathrm{localPredecessors}(v,i)\right| - 1$}
                    \State send $m=(\mathrm{\emph{build}}, \mathrm{localPredecessors}(v,i)[j+1])$\\\hspace{\algorithmicindent}\hspace{\algorithmicindent}\hspace{\algorithmicindent}\hspace{\algorithmicindent} to $\mathrm{localPredecessors}(v,i)[j]$
                \EndFor
                \For{$j=0 \text { to } \left|\mathrm{localSuccessors}(v,i)\right| - 1$}
                    \State send $m=(\mathrm{\emph{build}}, \mathrm{localSuccessors}(v,i)[j+1])$\\
                    \hspace{\algorithmicindent}\hspace{\algorithmicindent}\hspace{\algorithmicindent}\hspace{\algorithmicindent} to $\mathrm{localSuccessors}(v,i)[j]$
                \EndFor
            \EndFor
        \EndFunction
    \end{algorithmic}
    }}
    \caption[LinearizeNeighbors()]{\emph{LinearizeNeighbors()} introduces the neighbors to each other.}
    \label{alg:LinearizeNeighbors}
\end{figure}

The other three periodic actions introduce new neighbors to each other to
reach the desired topology. In the first function \emph{IntroduceNode()}, node
$v$ introduces itself to all of its neighbors to create backward edges (cf.
Fig. \ref{alg:IntroduceNode}). In the second function
\emph{IntroduceClosestNeighbors()}, node $v$ introduces the two direct
neighbors, the closest predecessor and the closest successor, in each level to
all other neighbors in this level (cf. Fig.
\ref{alg:IntroduceClosestNeighbors}). The last function
\emph{LinearizeNeighbors()} linearizes the neighborhood (cf. Fig.
\ref{alg:LinearizeNeighbors}), i.e., each neighbor is introduced to the
subsequent neighbor.

In addition to the periodic actions we have different reactive actions
triggered by the $received(m)$ guard. Depending on the message type of the
message $m$ received by a node $v$, the \emph{Build}, \emph{Remove} or
\emph{Lookup} operation is called (cf. Fig. \ref{alg:MessageHandling}).

\begin{figure}[htb]
    \fbox{\parbox{0.975\columnwidth}{
    \begin{algorithmic}\footnotesize
        \State {$received(m)$} $\rightarrow$
        \If {$m=(\mathrm{\emph{build}}, x)$}
            \State $Build(x)$
        \ElsIf{$m=(\mathrm{\emph{remove}}, x)$}
            \State $Remove(x)$
        \ElsIf {$m=(\mathrm{\emph{lookup}},x)$}
            \State $Lookup(x)$
        \EndIf
    \end{algorithmic}
    }}
    \caption[Message handling]{The incoming messages at node $v$ are handled.}
    \label{alg:MessageHandling}
\end{figure}

The \emph{Build()} function of a node $v$ must handle two cases (cf. Fig.
\ref{alg:Build}): The node $x$ which is given as argument is already in the
neighborhood or not. If it is already included, its information (especially
its current bandwidth) is updated. Then the \emph{CheckNeighborhood()}
operation is called to check if the modified node $x$ and all other nodes are
still needed for the node. In the case that the node $x$ is not yet in the
neighborhood, it is checked by the \emph{CheckNode()} function (cf. Fig.
\ref{alg:CheckNode}) if the node should be integrated in the neighborhood. If
the node is needed in one level, it will be included in the neighborhood and
the whole neighborhood will be checked if there is an unnecessary node now. If
the new node is not needed by the node, it will be forwarded to the next node
$w$ with the longest common prefix.

\begin{figure}[htb]
    \fbox{\parbox{0.975\columnwidth}{
    \begin{algorithmic}\footnotesize
        \Function{Build}{node $x$}
            \If{$x \in v.nh$}
                \State update neighbor information
                \State CheckNeighborhood()
            \Else
                \State $v.nh = v.nh \cup \{x\}$
                \If{CheckNode($x$) = $true$}
                    \State CheckNeighborhood()
                \Else
                    \State $v.nh = v.nh \backslash \{x\}$
                    \State send $m=(\mathrm{\emph{build}}, x)$\\
                    \hspace{\algorithmicindent}\hspace{\algorithmicindent}\hspace{\algorithmicindent}\hspace{\algorithmicindent} to node $\operatorname{argmax}_{w\in v.nh}{\mathrm{commonPrefix}(x,w)}$
                \EndIf
            \EndIf
        \EndFunction
    \end{algorithmic}
    }}
    \caption[Build()]{\emph{Build} handles an incoming \emph{build} message.}
    \label{alg:Build}
\end{figure}

With these self-stabilizing rules the desired topology can be reached and maintained in a finite number of steps.

\subsection{Correctness}
\label{sec:Correctness}

Next we show the convergence and the closure of the algorithm, which implies
that it is a self-stabilizing algorithm for the HSkip+ topology. Most of the
proofs are skipped due to space constraints, but we dissected the
self-stabilization process into sufficiently small pieces so that it is not
too difficult to verify them with the help of the protocols.

\subsubsection{Convergence}

Here, we prove the following theorem:

\begin{theorem}[Convergence] \label{theorem:Convergence}
If $G(t)=(V,E)$ is weakly connected at time $t$, then $G_e(t')=G^{HSkip+}$ for
some time $t'>t$ and all node information is correct.
\end{theorem}

\begin{proof}
To prove this theorem we will show different lemmas which lead to the
convergence. But first, we give an overview of the whole proof:
Starting with any weakly connected graph we first show that the network stays
connected over time (cf. Lem. \ref{lemma:WeaklyConnected}). Furthermore, we
show that all wrong information about nodes in the network will be removed
over time (cf. Lem. \ref{lemma:WrongInformation}). Having reached this state,
we will show the creation of the correct linked list at the bottom level (cf.
Lem. \ref{lemma:ListCreation}) and the maintenance of it (cf. Lem.
\ref{lemma:ListMaintenance}). The next step is the proof of the creation of
the HSkip+ topology at the bottom level (cf. Lem. \ref{lemma:HSkipCreation})
and its maintenance (cf. Lem. \ref{lemma:HSkipMaintenance}). By showing the
inductive creation of the HSkip+ topology at all levels (cf. Lem.
\ref{lemma:HSkipInduction}) we can finish the convergence proof.

Note that edges are only deleted if there already exist other edges concerning the desired topology. Otherwise, edges are
only added or delegated in a sense that a node $u$ holding an identifier of
$v$ may forward that identifier to one of its neighbors $w$, which also
preserves connectivity. Hence, we get:

\begin{lemma}[Weakly Connected] \label{lemma:WeaklyConnected}
If $G(t)=(V,E)$ is weakly connected at time $t$, then $G(t')$ is weakly
connected at any time $t'>t$.
\end{lemma}
\begin{proof}
To show the connectivity, we have to prove that for all edges $(x,y) \in E(t)$ there is the same edge in the next time step, $(x, y) \in E(t+1)$, or there is a path from $x$ to $y$ which uses other edges. If the edge is still available in the next time step, $G(t+1)$ obviously stays weakly connected, so we only have to consider the case that $(x,y) \notin E(t+1)$. Here we distinguish between two different cases: Either the removed edge was an explicit edge or it was an implicit edge at time $t$.

So let us first consider the case where $(x,y)$ is an explicit edge in the graph and therefore $y$ is in the neighborhood of $x$. In our algorithm there is only one operation which removes nodes from the neighborhood: The \emph{CheckNeighborhood()} function removes nodes which are not needed for the correct neighborhood. But in this case, a message with the removed node $y$ is sent by the current node $x$ to another node $v$ in its neighborhood (the one with the longest common prefix). Therefore, we have a new implicit edge $(v, y)$ and together with the existing explicit edge $(x, v)$ as $v$ is in the neighborhood of $x$, there exists a path from $x$ to $y$ over $v$.

We now look at the case where $(x,y)$ is an implicit edge at time $t$ which means that we have the reference to node $y$ in an incoming \emph{Build} message at node $x$. The message is handled by the \emph{Build()} operation and two cases can occur. On the one hand, the node $y$ can be useful in the neighborhood of $x$, then $y$ is integrated in the neighborhood and at time $t+1$, $y \in x.nh$. So we have an explicit edge instead of the implicit one. On the other hand, the node $y$ can be useless for the node $x$. Then $y$ is delegated to another node $v$ in its neighborhood (the one with the longest common prefix) and $x$ stays connected with $y$ over the node $v$ and the explicit edge $(x, v)$ ($v$ is in the neighborhood of $x$) as well as the new implicit edge $(v, y)$.
\end{proof}

Due to our FIFO delivery and fair message receipt assumption, all messages
initially in the system will eventually be processed. Moreover, the distance
(w.r.t. the sorted ordering of the nodes) of wrong information in the network can be shown to
monotonically decrease. Once it is equal to one, the periodic self-introduction
(cf. Figure 5) ensures that the node information is corrected, which implies
the following lemma.

\begin{lemma}[Wrong Information] \label{lemma:WrongInformation}
If $G(t)=(V,E)$ is weakly connected at time $t$, then there is a time point
$t'$ at which all node information in $G(t'')$ is correct for any $t''>t'$.
\end{lemma}

\begin{proof}
We assume, that there is a wrong information in the network. This wrong information can be contained in a message or in an internal variable of a node.
A message is only delegated a finite number of steps because it decreases the distance to its target in each step by choosing as next hop a node with a longer common prefix. Eventually, the information is integrated in the internal variables of a node, so we only have to consider this case. Therefore at a time $\hat{t}$, there is no wrong information in any message, but there can be wrong information in the internal variables of a node.

Assume that a node $x$ has wrong information about node $y$ at time $\hat{t}$. If node $x$ is connected to node $y$ by an explicit edge, eventually also node $y$ is connected to node $x$ because of the periodic call of \emph{IntroduceNode()}. Periodically, $y$ sends a \emph{build} message to $x$ in the \emph{IntroduceNode()} function. This message is handled by the \emph{Build()} operation in the first case ($y \in x.nh$). The information about $y$ in the internal variables of $x$ is updated and all wrong information is removed.

Furthermore, no new wrong information is produced in the network without any operations from outside, as only correct information is sent. Therefore at a time $t'>t$ we have only correct information in the whole network.
\end{proof}

From now on, let us only consider time steps in which all node information is
correct. Next to the edge set of the desired topology $E^{HSkip+}$ we define the edge set $E^{HSkip+_i}$ which contains all edges belonging to the topology at level $i$ and the edge set $E^{list_i}$ which contains all edges needed for a linked list at level $i$. Then it follows from arguments in
\cite{Jacob:2009:DPT:1582716.1582741} for the edge set $E^{list_0}$ at level
0:

\begin{lemma}[List Creation] \label{lemma:ListCreation}
If $G(t)=(V,E)$ is weakly connected, then $E_e(t') \supseteq E^{list_0}$ at
some time $t'>t$.
\end{lemma}

\begin{proof}
As the graph is weakly connected at time $t$, there exists an undirected path from $v$ to $w$ for all pairs $(v,w) \in V^2$. Now we look at two nodes $v$ and $w$ which should be direct neighbors: If $v$ and $w$ are already direct neighbors regarding their bandwidth values, then we only need one periodic message of the \emph{IntroduceNode()} operation and we have the linked list completed at this point.

\begin{figure}[htb]
   \centering
   \includegraphics[scale=.25]{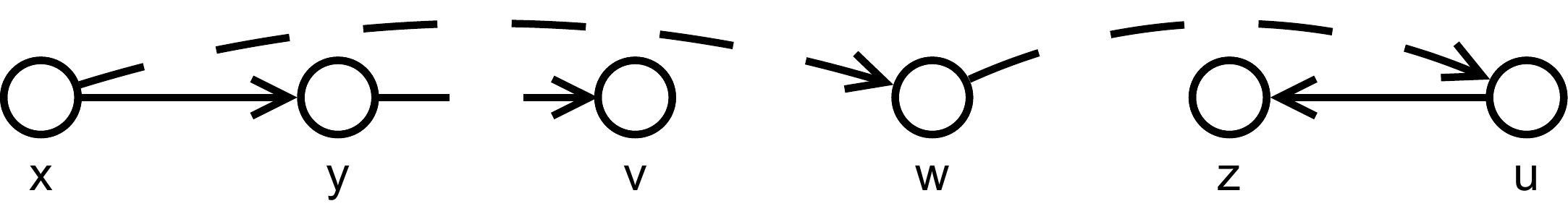}
   \caption[List Creation: Starting Path]{Path which connects $v$ with $w$ over intermediate nodes $y$, $x$, $u$, $z$.}
   \label{fig:ListCreation1}
\end{figure}

Let us now assume without loss of generality, that the path from $v$ to $w$ is through other nodes (cf. Fig. \ref{fig:ListCreation1}). We will show that the length of the path from $v$ to $w$ does not increase and finally becomes shorter that in the end $v$ and $w$ are directly connected. At first, we will show that the length of the path does not increase: For each edge $(x,y)$ in the path it yields, that the edge is only replaced by two edges which are in the range of $(x,y)$. We distinguish two cases:

\begin{figure}[htb]
   \centering
   \subfloat[][$x$ and $y$ are connected with an implicit edge $(x,y)$.]{
       \includegraphics[scale=.25]{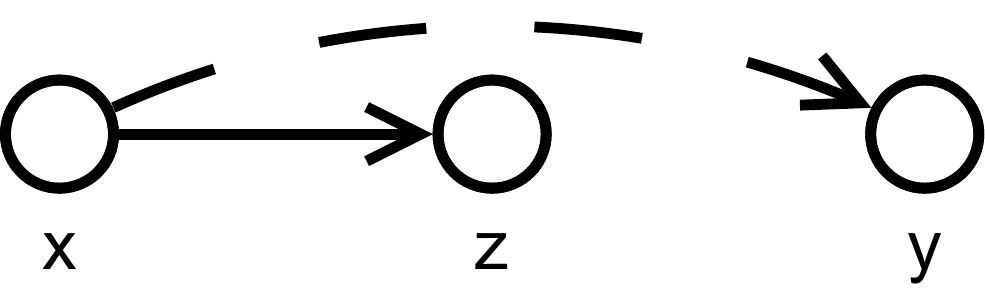}
       \label{subfig:ListCreation2}
   }
   \qquad
   \subfloat[][$x$ and $y$ are connected with an explicit edge $(x,y)$.]{
       \includegraphics[scale=.25]{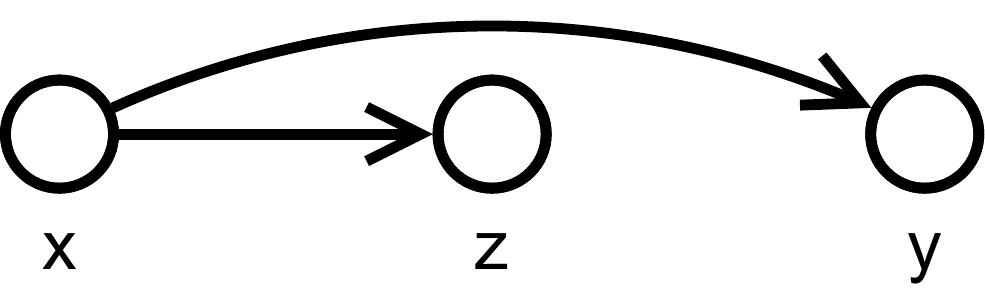}
       \label{subfig:ListCreation3}
   }

   \caption[List Creation: Range]{Edge $(x,y)$ is only replaced by two edges that are within the range of $(x,y)$.}
   \label{fig:ListCreationRange}
\end{figure}

\begin{enumerate}
   \item $(x,y)$ is an implicit edge (cf. Fig. \ref{subfig:ListCreation2}):\\
   If the edge $(x,y)$ is not needed by $x$, there has to be a closer node $z$ in the neighborhood of $x$. As consequence, $y$ can be delegated to $z$ and the edge $(x,y)$ is replaced by the two edges $(x,z)$ and $(z,y)$ which fulfills the condition to stay in the range.
   \item $(x,y)$ is an explicit edge (cf. Fig. \ref{subfig:ListCreation3}):\\
   The edge $(x,y)$ is only replaced if $x$ has a closer neighbor $z$ in its neighborhood. But than $y$ can be delegated to $z$ and the edge $(x,y)$ is replaced by the two edges $(x,z)$ and $(z,y)$ which fulfills the condition to stay in the range.
\end{enumerate}

If $x$ has no further outgoing edge $(x,z)$, the edge $(x,y)$ is kept or integrated as explicit one since $x$ does not know about closer neighbors.

We now show that border nodes in the path (the leftmost and rightmost nodes) will eventually be eliminated. Therefore, we assume that the path from $v$ to $w$ contains a border node $x$. Without loss of generality, we assume that $x$ has the maximum bandwidth in the path, $x.bw > v.bw > w.bw$. There are various cases, how $x$ is connected with two edges to the path: Explicit or implicit, the edges may point to $x$ or away from $x$. The undirected path from $v$ to $w$ contains somewhere the successive sequence of the nodes $y$, $x$ and $z$ (w.l.o.g. $x.bw > y.bw >  z.bw$). We want to exclude node $x$, so that the path only uses the nodes $y$ and $z$. We will now consider theses cases one by one:

\begin{figure}[htb]
   \centering
   \subfloat[][Explicit edge $(x,y)$ and explicit edge $(x,z)$.]{
       \includegraphics[scale=.25]{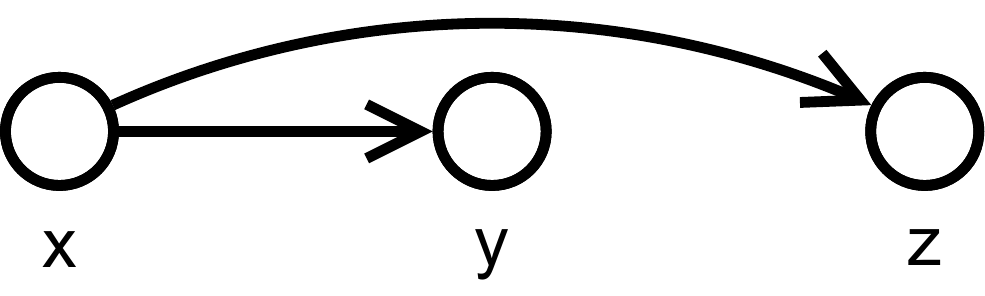}
       \label{subfig:ListCreation4}
   }
   \qquad
   \subfloat[][Explicit edge $(x,y)$ and implicit edge $(x,z)$.]{
       \includegraphics[scale=.25]{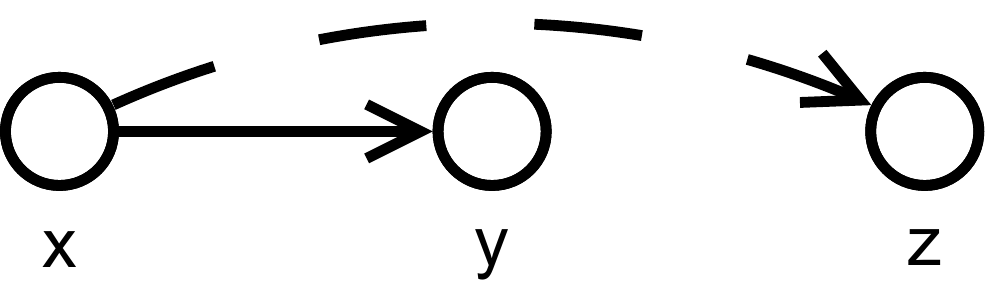}
       \label{subfig:ListCreation5}
   }
   \qquad
   \subfloat[][Implicit edge $(x,y)$ and explicit edge $(x,z)$.]{
       \includegraphics[scale=.25]{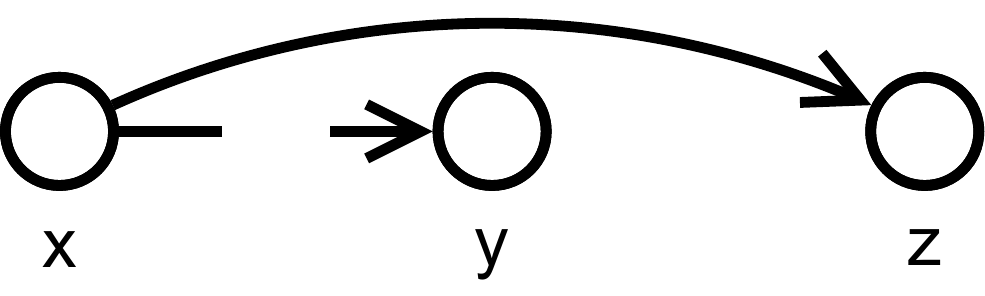}
       \label{subfig:ListCreation6}
   }
   \qquad
   \subfloat[][Implicit edge $(x,y)$ and implicit edge $(x,z)$.]{
       \includegraphics[scale=.25]{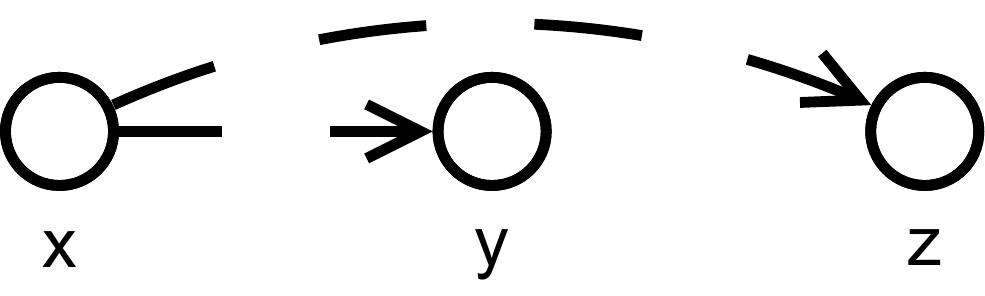}
       \label{subfig:ListCreation7}
   }
   \caption[List Creation: Border Node]{The border node $x$ can be excluded from the path.}
   \label{fig:ListCreationSimpleCases}
\end{figure}

   \begin{enumerate}
       \item[1a)] $x$ has explicit edges to $y$ and to $z$ (cf. Fig. \ref{subfig:ListCreation4}):\\
       With the \emph{LinearizeNeighbors()} operation of node $x$, a new implicit edge $(y, z)$ is added since $y$ and $z$ are successive neighbors of $x$. We have a direct connection between $y$ and $z$ and $x$ can be excluded.
       \item[1b)] $x$ has an explicit edge to $y$ and an implicit one to $z$ (cf. Fig. \ref{subfig:ListCreation5}):\\
       Node $x$ integrates $z$ in its neighborhood, as it needs at least two successors. The edge $(x,z)$ is converted to an explicit edge and we have case a.
       \item[1c)] $x$ has an implicit edge to $y$ and an explicit edge to $z$ (cf. Fig. \ref{subfig:ListCreation6}):\\
       As in the previous case, Node $x$ integrates $y$ in its neighborhood since it needs at least two successors. The edge $(x,y)$ is also converted to an explicit one and we have case a.
       \item[1d)] $x$ has implicit edges to $y$ and to $z$ (cf. Fig. \ref{subfig:ListCreation7}):\\
       Depending on the order of processing, $y$ or $z$ is integrated in the neighborhood of $x$ since it has no successor. As consequence this case is reduced to case b or c.
   \end{enumerate}

Additionally, we have to consider the cases in which there is another node $u$ to which node $x$ has an explicit edge to:

\begin{figure}[htb]
   \centering
   \subfloat[][$u$ is the closest successor of $x$.]{
       \includegraphics[scale=.25]{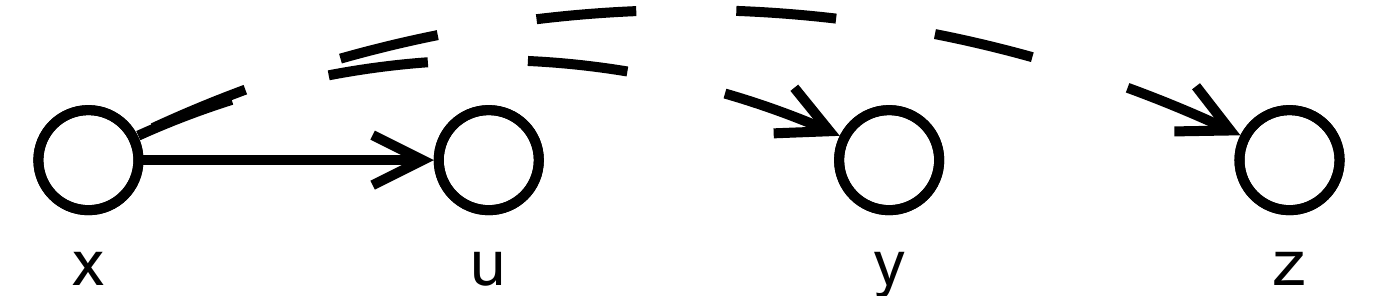}
       \label{subfig:ListCreation8}
   }
   \qquad
   \subfloat[][$u$ is the intermediate successor of $x$.]{
       \includegraphics[scale=.25]{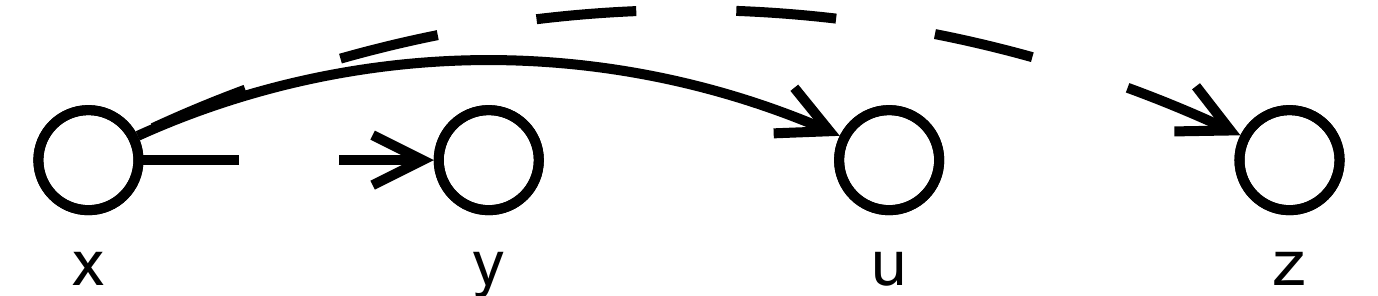}
       \label{subfig:ListCreation9}
   }
   \qquad
   \subfloat[][$u$ is the farthest successor of $x$.]{
       \includegraphics[scale=.25]{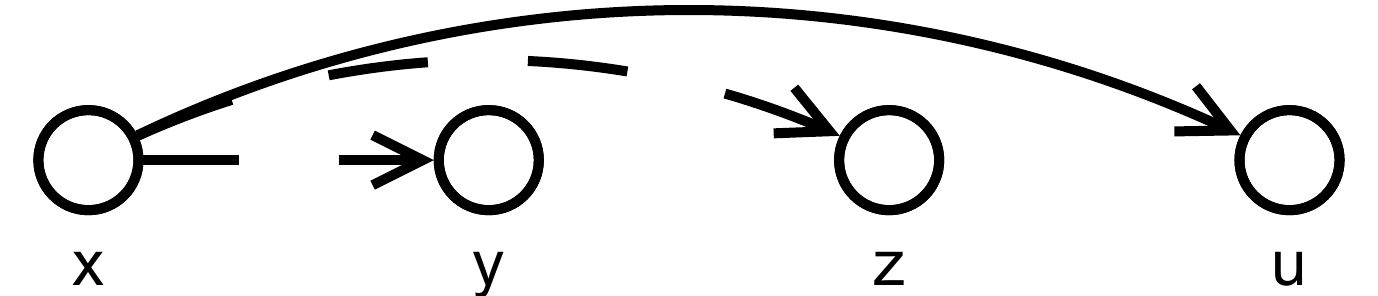}
       \label{subfig:ListCreation10}
   }
   \caption[List Creation: Border Node with Intermediate Node]{The border node $x$ can be excluded with the use of node $u$.}
   \label{fig:ListCreationComplexCases}
\end{figure}

   \begin{enumerate}
       \item[2a)] $x.bw > u.bw > y.bw > z.bw$ (cf. Fig. \ref{subfig:ListCreation8}):\\
       The results depends on the execution order: If the edge $(x,y)$ is considered first, $y$ is integrated in the neighborhood of $x$. In the next step with the \emph{LinearizeNeighborhood()} operation, a new implicit edge $(u,y)$ is added since $u$ and $y$ are successive neighbors of $x$. We can treat $u$ as new $y$ and the case 1b holds.\\
       If $(x,z)$ is considered first, $z$ is integrated in the neighborhood, the implicit edge $(u,z)$ is added with the \emph{LinearizeNeighborhood()} operation and with $u$ as new $z$, the case 1b also holds.
\item[2b)] $x.bw > y.bw > u.bw > z.bw$ (cf. Fig. \ref{subfig:ListCreation9}):\\
       The same argumentation as in case 2a reduces to case 1b or 1c.
\item[2c)] $x.bw > y.bw > z.bw > u.bw$ (cf. Fig. \ref{subfig:ListCreation10}):\\
       The same argumentation as in case 2a reduces to case 1b or 1c.
   \end{enumerate}

For the cases in which the edges point to $x$, a similar argumentation is possible to exclude the border node $x$. The exclusion of border nodes is executed for all border nodes on the undirected path from $v$ to $w$ until $v$ and $w$ are directly connected by an edge $(v,w)$ or $(w,v)$. Eventually, for each pair of direct neighbors $v$ and $w$ it yields that there is an edge $(v,w) \in E$ or $(w,v) \in E$. As the direct neighbors are always needed in our topology, the edges will be converted to explicit edges, if they are still implicit ones. By using one \emph{IntroduceNode()} operation, the backward edges will be created for each pair and also integrated as explicit edges. Since now the double-linked list is completed at the base level, $G_e(t') \supseteq G_e^{list_0}$ at time $t'>t$.
\end{proof}

Moreover, it follows from the algorithm that the set of linked lists at level $i$,
$E^{list_i}$, is maintained over time.

\begin{lemma}[List Maintenance] \label{lemma:ListMaintenance}
If $E_e(t) \supseteq E^{list_i}$ at time $t$, then $E_e(t') \supseteq
E^{list_i}$ at any time $t'>t$.
\end{lemma}

\begin{proof}
We start with a correct list $G_e^{list_i}$ at time $t$. As there are no external dynamics, no nodes enter or leave the system. To destroy the correct list, a node has to be removed from the neighborhood of any other node. The only removal happens in the \emph{CheckNeighborhood()} operation after the check if the node is needed. But this check always returns true for the closest neighbors of a node because they are always needed. They are always in the range at level $i$ and \emph{CheckNode()} returns $true$. Therefore no edges are removed and it follows that $G_e(t') \supseteq G_e(t) \supseteq G_e^{list_i}$.
\end{proof}

Starting with a linked list at a level $i$ we can show the creation of the
HSkip+ links at the same level, $E^{HSkip+_i}$:

\begin{lemma}[HSkip+ Creation] \label{lemma:HSkipCreation}
If $E_e(t) \supseteq E^{list_i}$ at time $t$, then eventually $E_e(t')
\supseteq E^{HSkip+_i}$ at time $t'>t$.
\end{lemma}

\begin{proof}

\begin{figure}[htb]
   \centering
   \subfloat[][Double-linked list in the environment of node $v$ at a level $i$.]{
       \includegraphics[scale=.25]{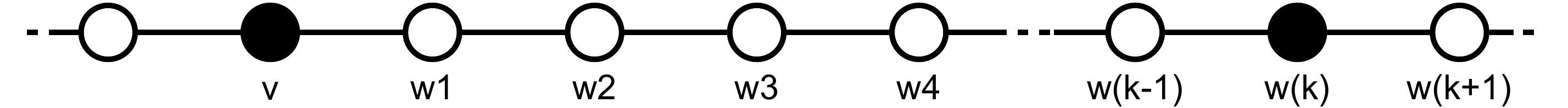}
       \label{subfig:HSkip+Creation1}
   }
   \qquad
   \subfloat[][Node $v$ and node $w_2$ are introduced each other by node $w_1$.]{
       \includegraphics[scale=.25]{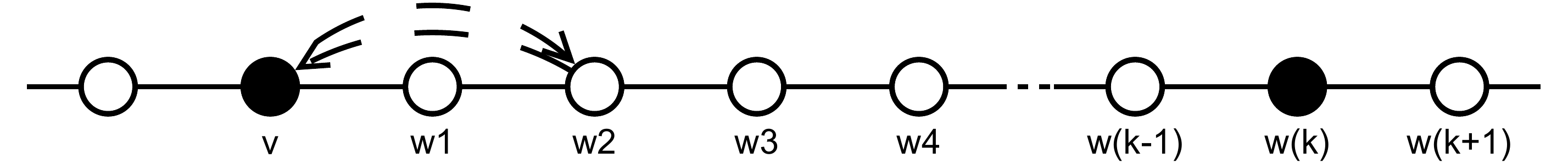}
       \label{subfig:HSkip+Creation2}
   }
   \qquad
   \subfloat[][Node $v$ and node $w_2$ are connected through skip edges.]{
       \includegraphics[scale=.25]{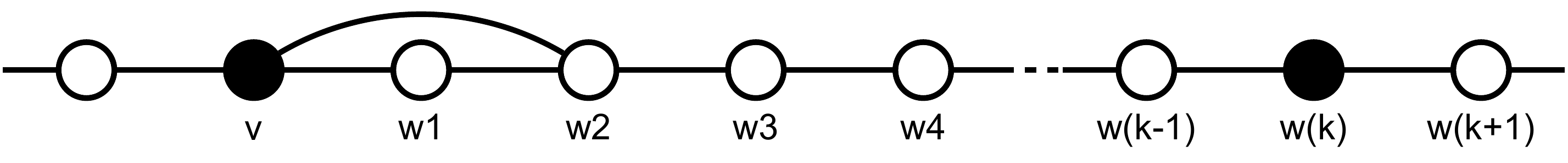}
       \label{subfig:HSkip+Creation3}
   }
   \qquad
   \subfloat[][Node $w_3$ is introduced to node $v$ by node $w_2$.]{
       \includegraphics[scale=.25]{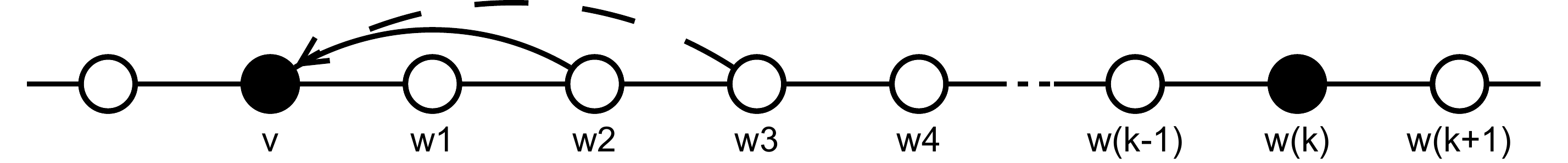}
       \label{subfig:HSkip+Creation4}
   }
   \qquad
   \subfloat[][Node $v$ is connected to node $w_3$ and introduces itself to node $w_3$.]{
       \includegraphics[scale=.25]{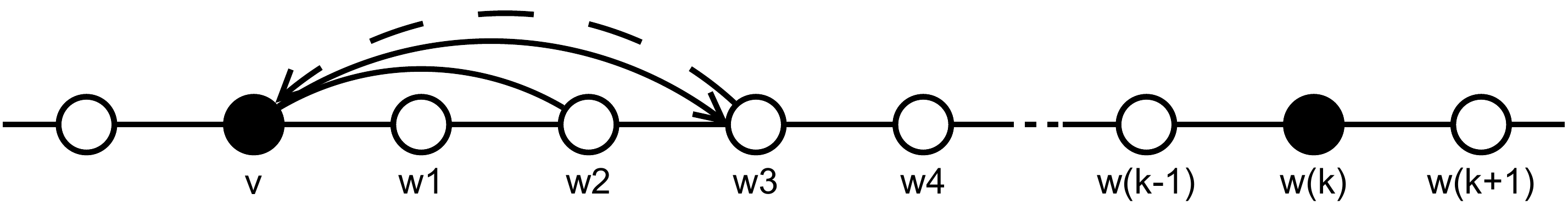}
       \label{subfig:HSkip+Creation5}
   }
   \qquad
   \subfloat[][Node $v$ and node $w_3$ are connected through skip edges.]{
       \includegraphics[scale=.25]{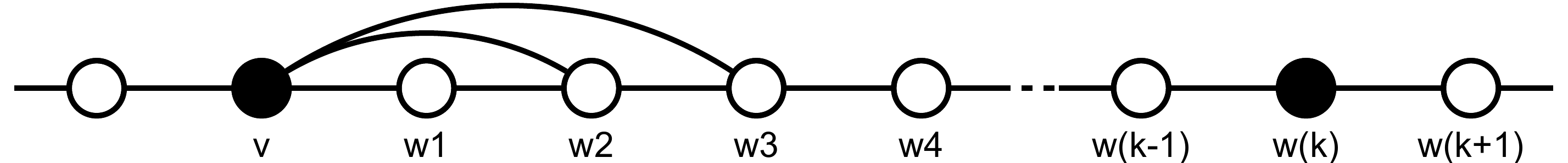}
       \label{subfig:HSkip+Creation6}
   }
      \qquad
      \subfloat[][Node $v$ is connected to nodes $w_1, \ldots w_k$; node $w_{k+1}$ is introduced to node $v$ by node $w_k$.]{
          \includegraphics[scale=.25]{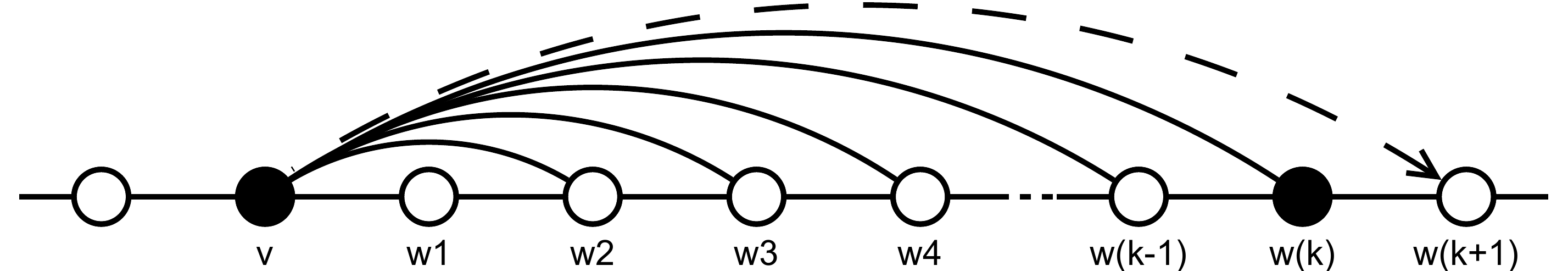}
          \label{subfig:HSkip+Creation7}
      }
   \caption[HSkip+ Creation at Level $i$]{The correct creation of the skip+ edges for node $v$ at level $i$.}
   \label{fig:HSkip+Creation}
\end{figure}

We start with a correct linked list $G_e^{list_i}$ which is created and obtained (cf. Lem. \ref{lemma:ListCreation} and \ref{lemma:ListMaintenance}). We look at an arbitrary node $v$. It is connected to its direct predecessor $u_1$ and successor $w_1$, formally $\mathrm{localClosestPred}(v,i) = \mathrm{closestPred}(v,i)$ and $\mathrm{localClosestSucc}(v,i) = \mathrm{closestSucc}(v,i)$ (cf. Fig. \ref{subfig:HSkip+Creation1}). The same holds for all other nodes, they are all connected to their direct predecessors and successors.

First we only look at the successors of node $v$. With the \emph{IntroduceClosestNeighbors()} operation, node $w_1$ introduces the nodes $v$ and $w_2$ to each other and the implicit edges $(v,w_2)$ and $(w_2,v)$ are created (cf. Fig. \ref{subfig:HSkip+Creation2}). As both edges are needed in the neighborhood of the nodes $v$ and $w_2$, they are both converted to explicit edges (cf. Fig. \ref{subfig:HSkip+Creation3}).

Now node $w_2$ introduces $w_3$ to $v$ (cf. Fig. \ref{subfig:HSkip+Creation4}). As the node $w_3$ is needed in the neighborhood of $v$ it is also converted to the explicit edge $(v, w_3)$ and with the \emph{IntroduceNode()} operation the implicit edge $(w_3, v)$ is created (cf. Fig. \ref{subfig:HSkip+Creation5}). Also this edge is converted to an explicit edge (cf. Fig. \ref{subfig:HSkip+Creation6}).

The same procedure is continued for the nodes $w_4, \ldots, w_k$ which are all needed neighbors for node $v$. Then the last neighbor $w_k$ introduces $w_{k+1}$ to node $v$ (cf. Fig. \ref{subfig:HSkip+Creation7}). As it is not needed, it is not integrated in the neighborhood. This implicit edge is also created in the future by the \emph{IntroduceClosestNeighbors()} operation of node $w_k$, but never converted to an explicit edge at level $i$. The successors are created correctly and the same observation can be made for the predecessors. Therefore, we eventually have all needed edges for the topology and $G_e(t) \supseteq G_e^{HSkip+_i}$.
\end{proof}

Additionally, it follows from our rules that the HSkip+ topology at a level
$i$ is maintained over time.

\begin{lemma}[HSkip+ Maintenance] \label{lemma:HSkipMaintenance}
If $E_e(t) \supseteq E^{HSkip+_i}$ at time $t$, then $E_e(t') \supseteq
E^{HSkip+_i}$ at any time $t'>t$.
\end{lemma}

\begin{proof}
The proof is equivalent to the list maintenance (cf. Lem. \ref{lemma:ListMaintenance}): Explicit edges are only removed during the \emph{CheckNeighborhood()} operation. As all neighbors are needed, no edge is removed and we get $G_e(t') \supseteq G_e(t) \supseteq G_e^{HSkip+_i}$.
\end{proof}

Using the simple observation that $E^{list_{i+1}} \subseteq E^{HSkip+_i}$, we
can then conclude:

\begin{lemma}[HSkip+ Induction] \label{lemma:HSkipInduction}
If $E_e(t) \supseteq E^{HSkip+_i}$ at time $t$, then $E_e(t') \supseteq
E^{HSkip+_{i+1}}$ at some time $t'>t$.
\end{lemma}

\begin{proof}

\begin{figure}[htb]
   \centering
   \subfloat[][The topology at level $i$ for one component.]{
       \includegraphics[scale=.25]{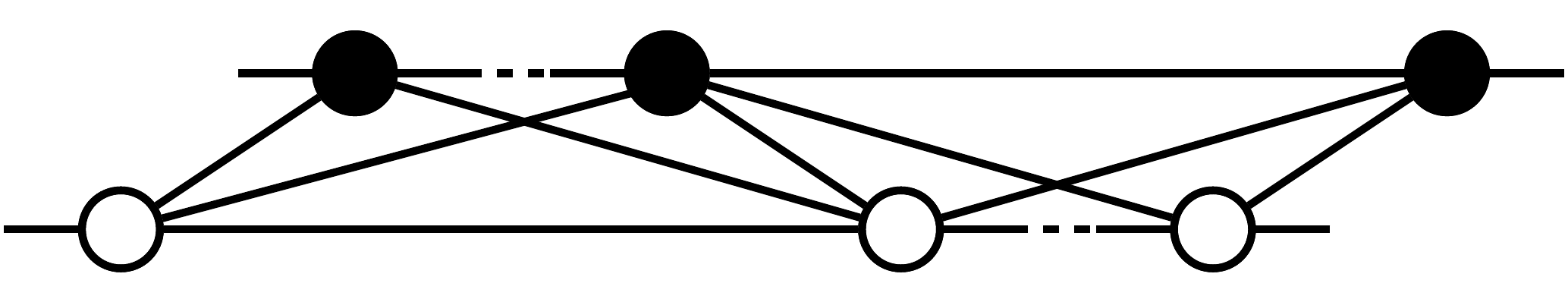}
       \label{subfig:HSkip+Induction1}
   }
   \qquad
   \subfloat[][The double-linked lists at level $i+1$.]{
       \includegraphics[scale=.25]{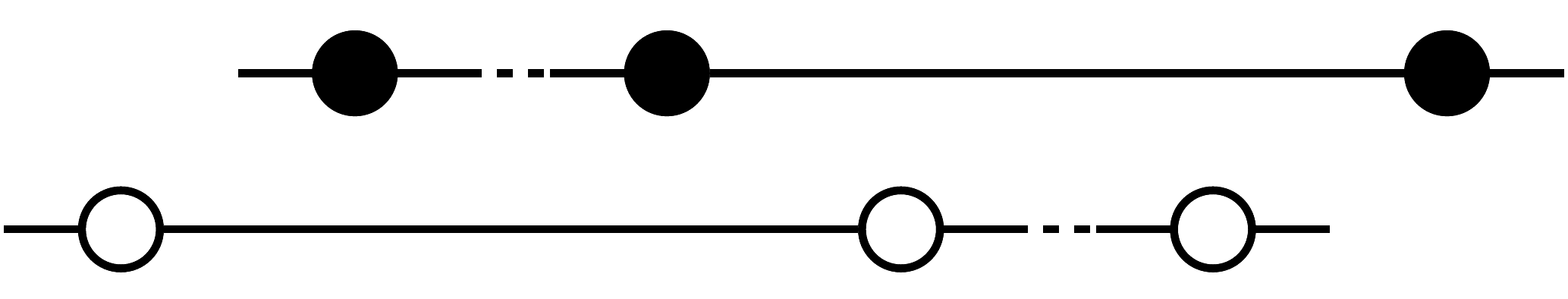}
       \label{subfig:HSkip+Induction2}
   }
   \caption[HSkip+ Creation at Level $i+1$]{The inductive step from level $i$ to level $i+1$.}
   \label{fig:HSkip+Induction}
\end{figure}

Starting with $G_e^{HSkip+_i}$, each node has an edge to the previous and one to next node with $0$ and $1$ as next bit of the bit string (cf. Fig. \ref{subfig:HSkip+Induction1}). Each component of level $i$ is separated into two components of level $i+1$ according to the $i+1$-th bit. Furthermore, the components of the next level are already connected and they already form a double-linked list (cf. Fig. \ref{subfig:HSkip+Induction2}). Now we can apply Lemma \ref{lemma:HSkipCreation} and the correct HSkip+ topology is formed at level $i+1$ and maintained over the time (cf. Lem. \ref{lemma:HSkipMaintenance}).
\end{proof}

The previous lemmas immediately imply Theorem \ref{theorem:Convergence}.
\end{proof}

\subsubsection{Closure}

After proving the convergence of our algorithm we need to show its closure.
This means that the network stays in a legal state once it has reached one.
Formally, we show (cf. Def. \ref{def:Closure}):

\begin{theorem}[Closure]
If $G_e(t)=G^{HSkip+}$ at time $t$, then $G_e(t')=G^{HSkip+}$ at any time
$t'>t$.
\end{theorem}

\begin{proof}
Suppose that at time $t$ we have a network which forms the desired topology
with its explicit edges, i.e., $E_e(t) = E^{HSkip+}$. If $E_e(t+1) \neq
E_e(t)$, then at least one explicit edge is added or removed. Edges are only
removed in the \emph{CheckNeighborhood()} operation if they are not needed and
edges are only added in the \emph{Build()} operation if they are needed for
the neighborhood. However, $G_e(t)$ already forms the correct neighborhood at
time $t$ at all levels (cf. Def. \ref{def:HSkip+}). As there are no external
changes and faults, the neighborhood is still correct at time $t+1$.
Therefore, no edge is removed or added, so $G_e(t+1) = G_e(t) = G^{HSkip+}$.
\end{proof}

The convergence and closure together show the correctness of the
self-stabilizing algorithm. In other words we have designed an algorithm which
creates the HSkip+ topology and ensures that it stays correct.

\subsection{External Dynamics}
\label{sec:ExternalDynamics}

As external dynamics of a network we consider all events which can have an
influence on the network. Two typical events for a peer-to-peer system are
arrivals and departures of nodes. The network has to adapt the topology in
this case. In our HSkip+ network we will also consider changes in the
bandwidths of the nodes because also this will have an influence on the
network topology. To handle these events we need a Join, Leave, and Change
operation.

\subsubsection{Join}
\label{subsec:Join}

The \emph{Join} operation in our network is very simple. If a node $v$ wants
to join the network, it just has to introduce itself to some node $w$. The
rest will be handled by the self-stabilization. Hence, it suffices to execute
the code in Fig.~\ref{alg:Join}.

\begin{figure}[htb]
    \fbox{\parbox{0.975\columnwidth}{
    \begin{algorithmic}\footnotesize
        \Function{Join}{}
            \State send $m=(\mathrm{\emph{build}}, v)$ to a known node $w$
        \EndFunction
    \end{algorithmic}
    }}
    \caption[Join()]{The \emph{Join} operation of node $v$ sends a \emph{build} message to $w$.}
    \label{alg:Join}
\end{figure}

\subsubsection{Leave}
\label{subsec:Leave}

We distinguish between two cases: A scheduled \emph{Leave} and a \emph{Leave}
caused by a failure. In the first case (cf. Fig. \ref{alg:Leave}) the node $v$
which wants to leave the network simply sends a \emph{remove} message to all
of its neighbors. The receivers of these messages remove the node $v$ from
their neighborhood (cf. Fig. \ref{alg:Remove}). The leaving node also deletes
its entire neighborhood.

\begin{figure}[htb]
    \fbox{\parbox{0.975\columnwidth}{
    \begin{algorithmic}\footnotesize
        \Function{Leave}{}
            \ForAll{$w \in v.nh$}
                \State send $m=(\mathrm{\emph{remove}}, v)$ to node $w$
            \EndFor
            \State $v.nh = \emptyset$
        \EndFunction
    \end{algorithmic}
    }}
    \caption[Leave()]{ The \emph{Leave()} operation of node $v$ sends \emph{remove} messages.}
    \label{alg:Leave}
\end{figure}

The \emph{Remove()} operation removes a given node from the neighborhood if it
is present (cf. Fig. \ref{alg:Remove}). It is executed as reaction to a
\emph{remove} message. A neighborhood check is not needed, as there cannot be
too much information after removing some information.

\begin{figure}[htb]
    \fbox{\parbox{0.975\columnwidth}{
    \begin{algorithmic}\footnotesize
        \Function{Remove}{node $x$}
            \If{$x \in u.nh$}
                \State $u.nh = u.nh \backslash \{x\}$
            \EndIf
        \EndFunction
    \end{algorithmic}
    }}
    \caption[Remove()]{The \emph{Remove()} operation removes the node $x$ from $v.nh$.}
    \label{alg:Remove}
\end{figure}

Additionally, we can have a \emph{Leave} in the network caused by a node
failure. In this case, we assume the existence of a failure detector at the
nodes which checks periodically the existence of the neighbor nodes. Therefore, the outcome of a failed node is the same as for a scheduled
\emph{Leave}: all neighboring nodes invalidate their links to the failed node.

\subsubsection{Change}
\label{subsec:Change}

The \emph{Change} operation updates the bandwidth value of a node $v$ and
therefore the order in our topology has to be updated. The operation itself
only needs to update its internal variable of the bandwidth (cf. Fig.
\ref{alg:Change}).

\begin{figure}[htb]
    \fbox{\parbox{0.975\columnwidth}{
    \begin{algorithmic}\footnotesize
        \Function{Change}{new bandwidth $bw$}
            \State $v.bw = bw$
        \EndFunction
    \end{algorithmic}
    }}
    \caption[Change()]{The \emph{Change()} operation updates the bandwidth value of $v$.}
    \label{alg:Change}
\end{figure}

The correct recovery of the topology after no more external dynamics are
happening follows directly from the convergence.

\begin{theorem}[Recovery from External Dynamics] \label{theorem:External_Dynamics_Correctness}
If $G_e(t) = G^{HSkip+}$ at time $t$ and a node $v$ joins or a node $u$ leaves
or a node $u$ changes its bandwidth, then eventually $G_e(t') = G^{HSkip+}$ at
time $t' > t$.
\end{theorem}

\begin{proof}
  If the node $v$ joins the network by sending a \emph{build} message to some node $w$, the graph has a new node $v$ and a new implicit edge $(w,v)$. Therefore $G(t+1) = \left(V(t+1) = V(t) \cup \{v\}, E(t+1) = E(t) \cup \{(w,v)\}\right)$. $G(t+1)$ is obviously weakly connected. As there are no further external dynamics, we can apply the theorem about the convergence of our self-stabilizing algorithm (cf. The. \ref{theorem:Convergence}). After a finite number of steps, we eventually reach the correct topology, formally $G_e (t') = G^{HSkip+}$ at a time $t'>t$.
\end{proof}

It is easy to see that the worst case number of structural changes in a level
does not exceed $O(\log^2 n)$ w.h.p., but with more refined arguments one can
also show the following result:

\begin{theorem}[Structural Changes after External Dynamics] \label{theorem:Join_Changes}
If $G_e(t) = G^{HSkip+}$ at time $t$ and a node $v$ joins or a node $u$ leaves
or a node $u$ changes its bandwidth, then $G_e(t') = G^{HSkip+}$ after
$O(\log^2 n)$ structural changes, w.h.p.
\end{theorem}

\begin{proof}
  To analyze the structural changes which are needed after node $v$ joins the network, we will look at the obsolete edges which have to be removed and  at the new edges which need to be created during the self-stabilization.

  As the degree of the network is limited by $O(\log n)$, each node $w$ has only $O(\log n)$ neighbors in $w.nh$. So each node can only remove at most $O(\log n)$ edges. Furthermore only $O(\log n)$ nodes can be affected from any changes as the new node $v$ can be in at most $O(\log n)$ neighborhoods and cause changes. Therefore at most $O(\log^2 n)$ edges can be removed.

  In addition, the new node $v$ has $O(\log n)$ other nodes in its neighborhood. As already stated in the removing part, only $O(\log n)$ other nodes are affected by the \emph{Join}, each one can create at most $O(\log n)$ new edges. Therefore $O(\log n) + O(\log^2 n)$ new edges can be created.

  Altogether we have at most $O(\log^2 n)$ structural changes in the network after a \emph{Join}.
\end{proof}

With some effort, one can also show that this is an upper bound for the number
of additional messages (i.e., messages beyond those periodically created by
the $true$ guard).

\begin{theorem}[Workload of External Dynamics] \label{theorem:Join_Workload}
If $G_e(t) = G^{HSkip+}$ at time $t$ and a node $v$ joins or a node $u$ leaves
or a node $u$ changes its bandwidth, then $G_e(t') = G^{HSkip+}$ after
$O(\log^2 n)$ additional messages, w.h.p.
\end{theorem}

\begin{proof}
  Firstly, the \emph{build} message for the \emph{Join} from node $v$ to node $w$ has to be forwarded to a node $x$ which can include the new node $v$ in its neighborhood. This needs at most $O(\log n)$ additional messages because the diameter of the network is at most $O(\log n)$ and the message will be forwarded closer to the target in each step as the bit string is adjusted at least one bit in each step.

  At the  position we need several further \emph{build} messages to introduce and linearize the neighborhood. We have to create at most $O(\log^2 n)$ structural changes (cf. Lem. \ref{theorem:Join_Changes}). For each new edge, a message is created and integrated as an explicit edge in the next step. Therefore with $O(\log^2 n)$ messages we can create all new edges. Altogether we have a workload of at most $O(\log^2 n)$ additional messages to complete the \emph{Join} operation.
\end{proof}

\subsection{Routing}
\label{sec:Routing}

The routing of a message to a target node $x$ is handled by the message
$m=(\text{\emph{lookup}}, x)$. The \emph{Lookup()} operation handles the
forwarding of \emph{lookup} messages (cf. Alg. \ref{alg:Lookup}). If the
lookup target is equal to the current node, the lookup has finished. This is
checked with the help of the identifiers of the nodes. If this is not the
case, the lookup is forwarded to the next better node.

\begin{figure}[htb]
    \fbox{\parbox{.975\columnwidth}{
    \begin{algorithmic}\footnotesize
        \Function{Lookup}{node $x$}
            \If{$v.id=x.id$}
                \State done
            \Else
                \State $length$\\
                \hspace{\algorithmicindent}\hspace{\algorithmicindent}\hspace{\algorithmicindent}$\leftarrow \min{(\mathrm{level}(v), \mathrm{commonPrefix}(v, x))}$
                \If{$\mathrm{localFarthestPred}$\\
                        \hspace{\algorithmicindent}\hspace{\algorithmicindent}\hspace{\algorithmicindent}\hspace{\algorithmicindent}$(v, length, x.rs[length]) \neq nil$}
                    \State $w \leftarrow \mathrm{localFarthestPred}$\\
                        \hspace{\algorithmicindent}\hspace{\algorithmicindent}\hspace{\algorithmicindent}\hspace{\algorithmicindent}\hspace{\algorithmicindent}$(v, length, x.rs[length])$
                \Else
                    \State $w \leftarrow \mathrm{localFarthestSucc}$\\
                        \hspace{\algorithmicindent}\hspace{\algorithmicindent}\hspace{\algorithmicindent}\hspace{\algorithmicindent}\hspace{\algorithmicindent}$(v, length, x.rs[length])$
                \EndIf
                \State send $m=(\mathrm{\emph{lookup}}, x)$ to node $w$
            \EndIf
        \EndFunction
    \end{algorithmic}
    }}
    \caption[Lookup()]{The \emph{Lookup()} operation of $v$ handles the \emph{build} messages.}
    \label{alg:Lookup}
\end{figure}

The next hop is determined by the random bit strings of the involved nodes: In
each step we want to have at least one bit more in common with the target node
$x$. Formally, if at the current node $v$ the bit string has $i$ bits in
common with the target ($\mathrm{commonPrefix}(v, x) = i$), at the next step
at node $w$ we want at least $i+1$ bits in common ($\mathrm{commonPrefix}(w,
x) \geq i+1$). In our topology such a node is always present in level $i$ of a
node $v$ in both the successors and the predecessors (unless there is no such
higher level) as every node is connected at every level $i$ to at least one
predecessor and successor with $0$ and one with $1$ as next bit. First, the routing is done along predecessors so that the message
is guaranteed to follow a sequence of nodes of monotonically increasing
bandwidth. Once the node with the highest bandwidth defined by this rule has been reached, we use the successor nodes until the target node is reached. From the routing protocol it
follows:

\begin{theorem}[Correctness of Routing] \label{the:Routing_Corretness}
If $G_e = G^{HSkip+}$, then a message $m=(\mathrm{\emph{lookup}}, w)$ sent by
a node $u$ eventually reaches node $w$.
\end{theorem}

\begin{proof}
We prove the correctness inductive such that at each step the random bit string $v.rs$ of the current node $v$ is more similar to the target bit string $x.rs$. As inductive hypothesis it holds that at step $i$ at node $v$: $\mathrm{prefix}(v, i) = \mathrm{prefix}(x, i)$. As induction base we look at step $i=0$ at the source node $s$, here at least the prefix of length $0$ should be equal to the target bit string, which obviously is the case. In the induction step we start from node $v$ with an equal prefix of length $i$ and the message has to reach a node $w$ with an equal prefix of length $i+1$. In the correct HSkip+ topology, each node has a predecessor and a successor with $0$ and $1$ as next bit. Only at the border nodes, the successor or predecessor could be absent, but at least one of the sides is complete. Therefore the length of the equal prefix is extended. As the target node $x$ exists in the network and is included in the HSkip+ topology it will eventually be reached by the message.
\end{proof}

We can also guarantee the following critical property, this is important to keep the congestion low.

\begin{lemma}[Involved Nodes in Routing] \label{lem:InvolvedNodes}
If $G_e = G^{HSkip+}$, the routing from some node $u$ to $w$ only uses nodes
$v$ with $v.bw \ge \min{\{u.bw, w.bw\}}$.
\end{lemma}

The {\em dilation} is defined as the longest path which is needed to route a
packet from an arbitrary source to an arbitrary target node. Since it just
takes a single hop to climb up one level and there are at most $O(\log n)$
levels w.h.p., we get:

\begin{theorem}[Dilation of Routing] \label{the:Routing_Dilation}
In HSkip+ the dilation is at most $O(\log n)$ w.h.p.
\end{theorem}

\begin{proof}
   In each routing step we use the edges of a higher level to forward the message to the next hop. As we have at most $O(\log n)$ levels, after $O(\log n)$ steps we have reached the highest level. At each step we have adjusted the bit string to reach the target bit string. At the highest level we have connections to all nodes with the same bit string and reach our target node directly. If this were not the case a further level would be included in the network because there would be a non trivial component (cf. Def. \ref{def:Level}). Therefore, we need $O(\log n)$ steps to route a message from an arbitrary source to an arbitrary target and our dilation is $O(\log n)$.
\end{proof}

Finally, we analyze the congestion of our routing strategy. The {\em
congestion} of a node $v$ is equal to the total volume of the messages passing
it divided by its bandwidth. Let us consider the following routing problem:
Each node $u$ selects a target bitstring $x$ independently and uniformly at
random and sends a message of volume $\min\{u.bw,w.bw\}$ to the node $w$ with
the longest prefix match with $x$. Of course, $u$ may not know that volume in
advance. Our goal will just be to show that the expected congestion of this
routing problem is $O(\log n)$. If this is the case, then one can show similar
to \cite{Bhargava:2004:PDO:1007912.1007938} that the expected congestion of
routing any routing problem in HSkip+ is by a factor of at most $O(\log n)$
worse than the congestion of routing that routing problem in any other
topology of logarithmic degree, which matches the result in
\cite{Bhargava:2004:PDO:1007912.1007938}.

\begin{theorem}[Congestion of Routing] \label{the:Routing_Congestion}
In the HSkip+ topology, we can route messages with the presented routing
algorithm (cf. Alg. \ref{alg:Lookup}) and the defined routing problem with an
expected congestion of $O(\log n)$.
\end{theorem}

To prove the congestion we divide the routing process of a message from an
arbitrary source node $u$ to an arbitrary target node $w$ into two parts: On
the one hand the routing from $u$ along predecessor edges to an intermediate
node $v$ which has no better predecessors and on the other hand the routing
back over successors from node $v$ to the target node $w$. We will start with
the first part of the routing process, and initially we separately look at
each level $i$ in the network at node $v$. Later, we will sum it up for all
levels.

\begin{figure*}[tb]
\vspace{-0.25cm}
\begin{minipage}[b]{0.48\linewidth}
\centering
\includegraphics[width=\textwidth]{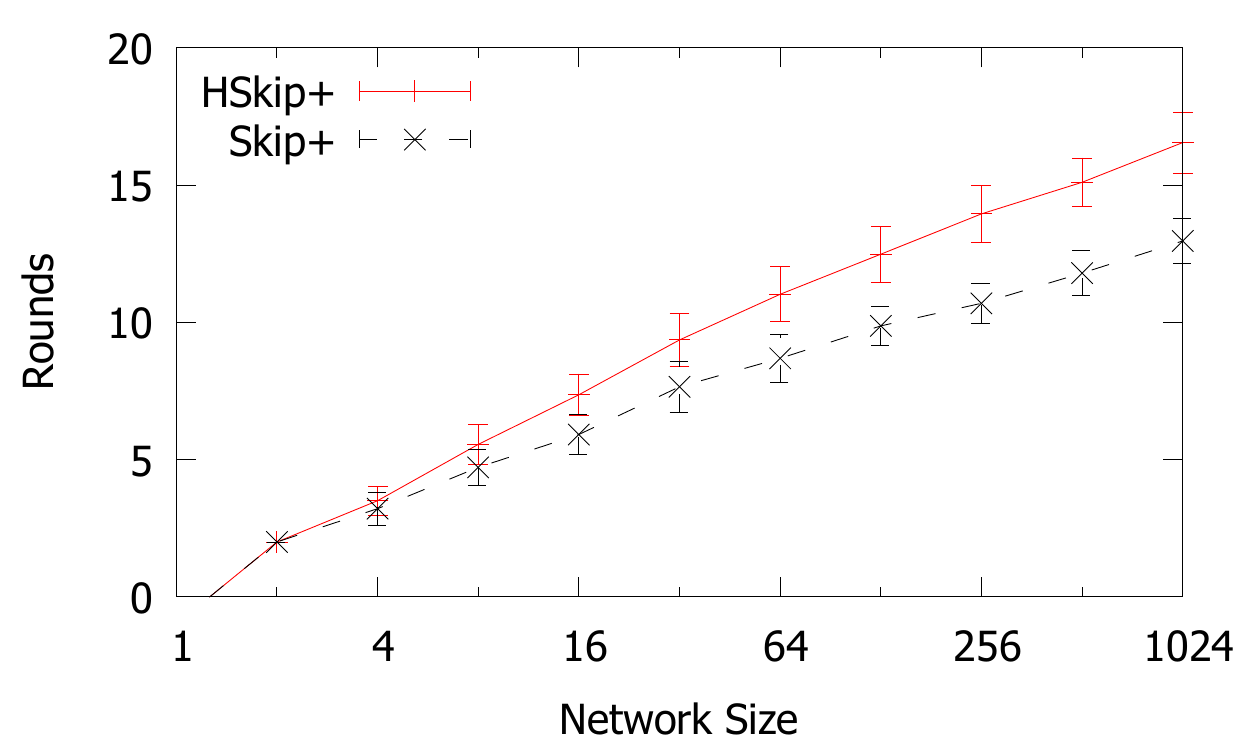}
\vspace{-0.75cm}
	\caption[Stabilization Time]{Stabilization Time.}
	\label{fig:Stabilization}
	
\end{minipage}
\hspace{0.5cm}
\begin{minipage}[b]{0.48\linewidth}
\centering
\includegraphics[width=\textwidth]{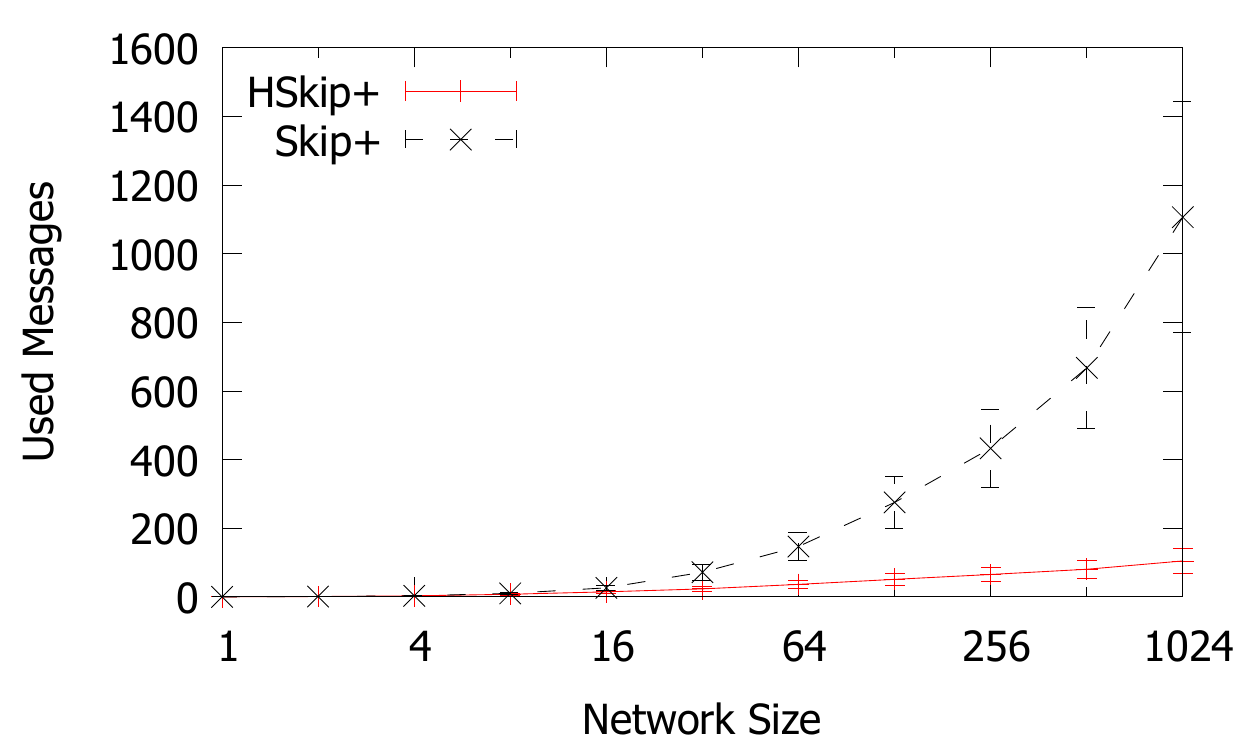}
\vspace{-0.75cm}
	\caption[Overall Used Messages Total]{Overall Used Messages Total.}
	\label{fig:UsedMessages}
\end{minipage}
\vspace{-0.5cm}
\end{figure*}

\begin{figure*}[tb]
\vspace{-0.25cm}
\begin{minipage}[b]{0.48\linewidth}
\centering
\includegraphics[width=\textwidth]{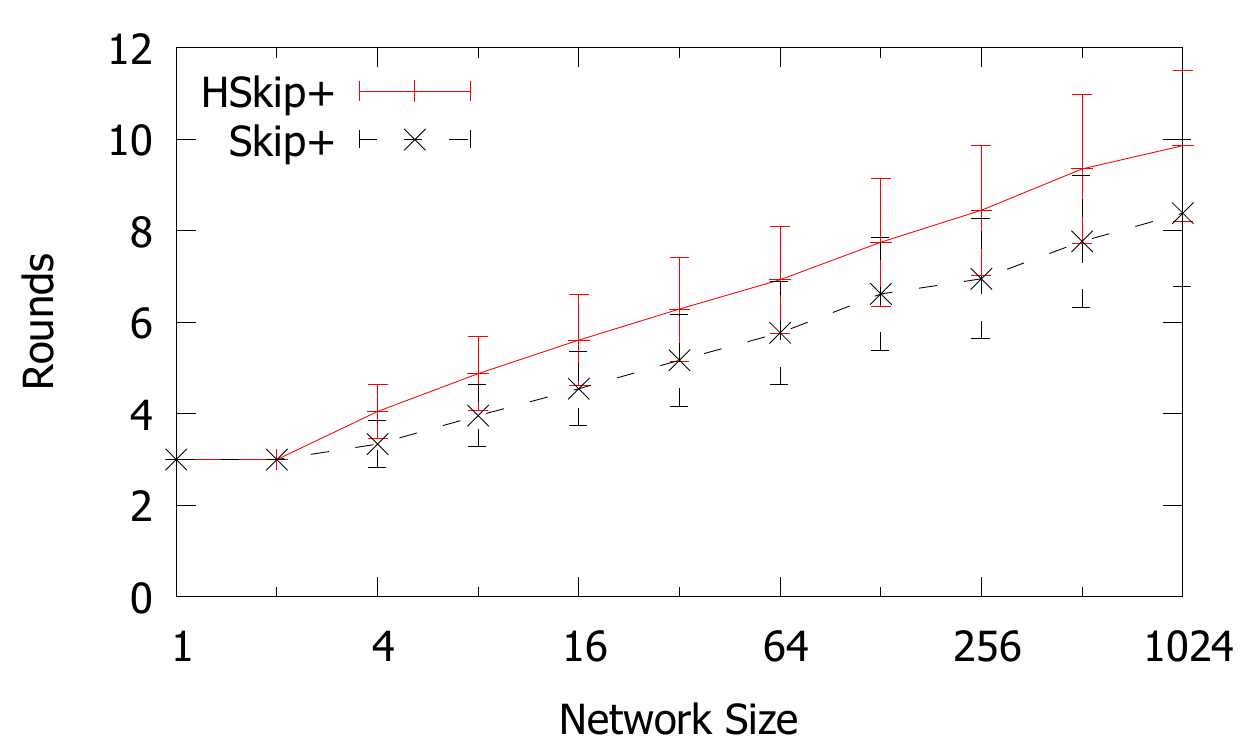}
\vspace{-0.75cm}
	\caption[Stabilization Time after Join]{Stabilization Time after Join.}
	\label{fig:JoinStabilization}
\end{minipage}
\hspace{0.5cm}
\begin{minipage}[b]{0.48\linewidth}
\centering
\includegraphics[width=\textwidth]{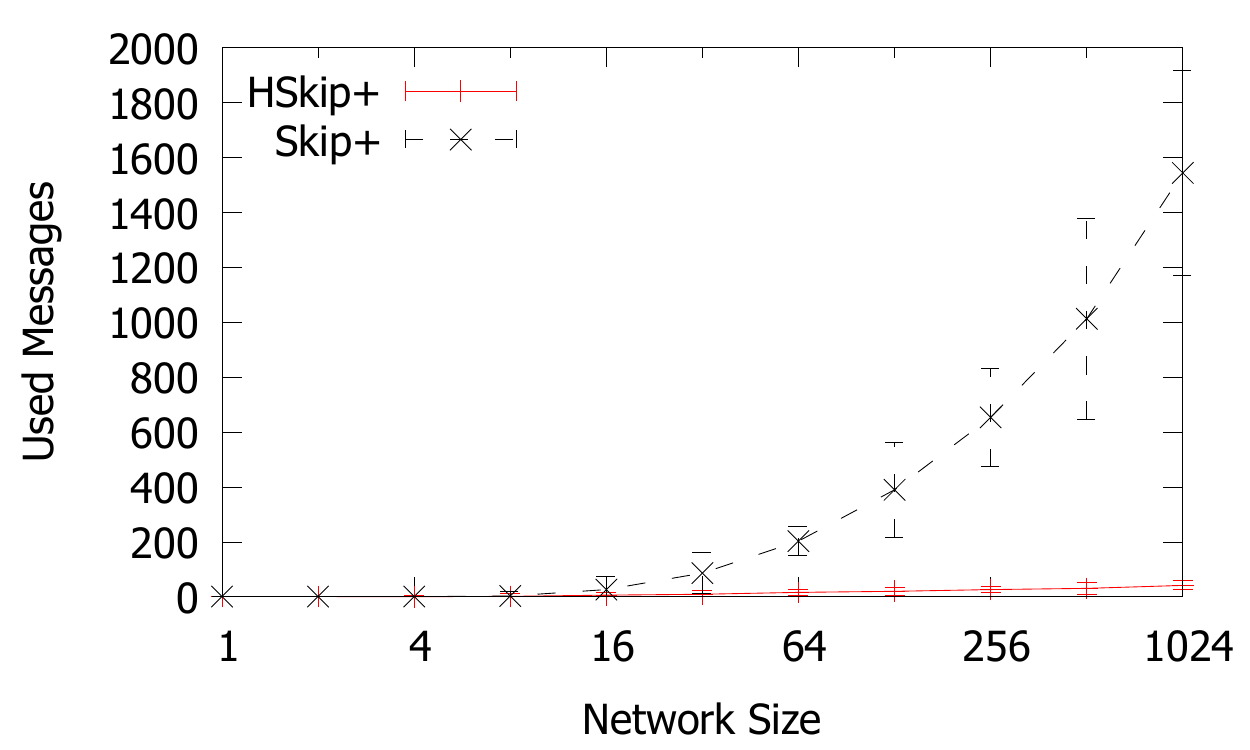}
\vspace{-0.75cm}
	\caption[Used Messages after Join]{Used Messages after Join.}
	\label{fig:JoinUsedMessages}
\end{minipage}
\vspace{-0.5cm}
\end{figure*}

\begin{lemma}[Congestion of node $v$ at level $i$] \label{lemma:RoutingCongestionLeveli}
In our routing problem (cf. The. \ref{the:Routing_Congestion}), every node $v$
has an expected congestion of $O(1)$ at every level $i \ge 1$.
\end{lemma}

\begin{proof}
We assume w.l.o.g. that the first $i$ bits of $v.rs$ are $1$. So messages with
the target bitstring starting with $1\ldots 1$ as the first $i$ bits will be
sent through the node $v$ as long as they start at a node $w$ that, at level
0, is between $v$ and the closest successor $w'$ of $v$ whose first $i$ bits
are also 1. We calculate the probability that there are $m$ nodes between $v$
and $w'$. For each node in between we have the probability
$(1-(\frac{1}{2})^i)$ that one of its first $i$ bits differs from 1, and for
$w'$ we get a probability of $(\frac{1}{2})^i$ that its first $i$ bits are 1.
Altogether, we get a probability that there are $m$ sources that may send
their messages to $v$ of $(1-(\frac{1}{2})^i)^m \cdot (\frac{1}{2})^i =
\frac{(2^i - 1)^m}{(2^i)^{m+1}}$.
Furthermore, we have to estimate the fraction of messages which will be sent
through $v$. Each message is sent through $v$ if the first $i$ bits are $1$,
which has a probability of $\frac{1}{2^i}$. Since we have $m$ possible sources
(excluding $v$), a fraction of $(m+1) \cdot \frac{1}{2^i}$ messages will be
sent through $v$. Lastly, we have to look at the volume sent by a single
source. We know that a node $s$ can send at most a volume of $s.bw$.
As we have only successors of $v$ as sources and for all successors $s$ it
holds that $s.bw < v.bw$, we can upper bound the volume sent by each source by
$v.bw$. Altogether, if we have $m$ nodes between $v$ and $w'$, the total
expected volume through $v$ at level $i$ is at most $(m+1) \cdot \frac{1}{2^i}
\cdot v.bw$.
Let the random variable $X_v^i$ be the total volume sent through $v$ at level
$i$. Then
\[
  E\left[X_v^i\right] = \sum_{m=0}^{\infty}{\frac{(2^i-1)^m}{(2^i)^{m+1}} \cdot (m+1)
  \cdot \frac{1}{2^i} \cdot v.bw} = O(v.bw)
\]
\end{proof}

Summing up the congestion over all levels, we get:

\begin{lemma}[Congestion of node $v$] \label{lemma:RoutingCongestionNode}
In our routing problem (cf. The. \ref{the:Routing_Congestion}), an arbitrary
node $v$ has an expected congestion of $O(\log{n})$.
\end{lemma}

\begin{proof}[Proof of Theorem \ref{the:Routing_Congestion}]
With the previous lemmas (cf. Lem. \ref{lemma:RoutingCongestionLeveli} and
Lem. \ref{lemma:RoutingCongestionNode}) we have calculated the congestion
caused by the first part of the routing process.
For the second part, we can argue in a similar way by looking at the routing
path backwards from the target node $x$ to the intermediate node $v$ and we
also get an expected volume of $O(\log n) \cdot v.bw$ at node $v$. Both routing parts together yield an expected congestion of $O(\log n) $ at
every node $v$.
\end{proof}

\begin{figure*}[tb]
\begin{minipage}[b]{0.48\linewidth}
\centering
\includegraphics[width=\columnwidth]{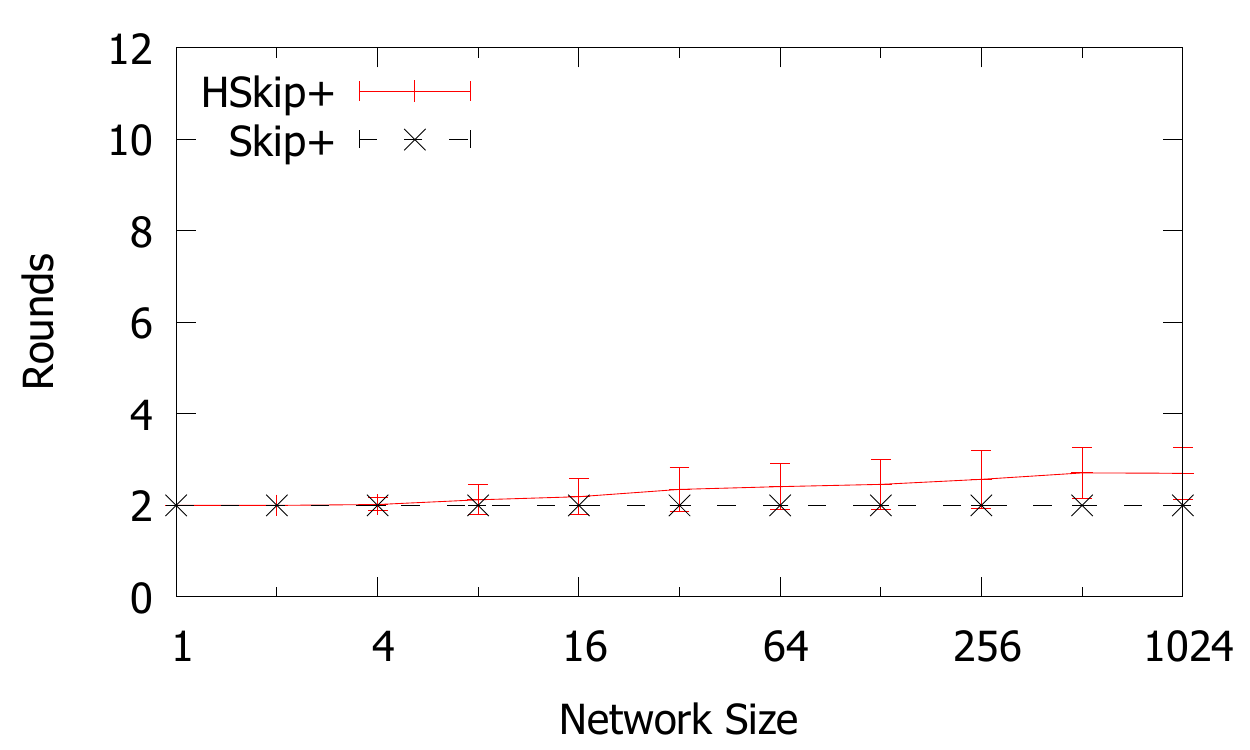}
	\vspace{-0.75cm}
	\caption[Stabilization Time after Leave]{Stabilization Time after Leave.}
	\label{fig:LeaveStabilization}
\end{minipage}
\hspace{0.5cm}
\begin{minipage}[b]{0.48\linewidth}
\centering
\includegraphics[width=\columnwidth]{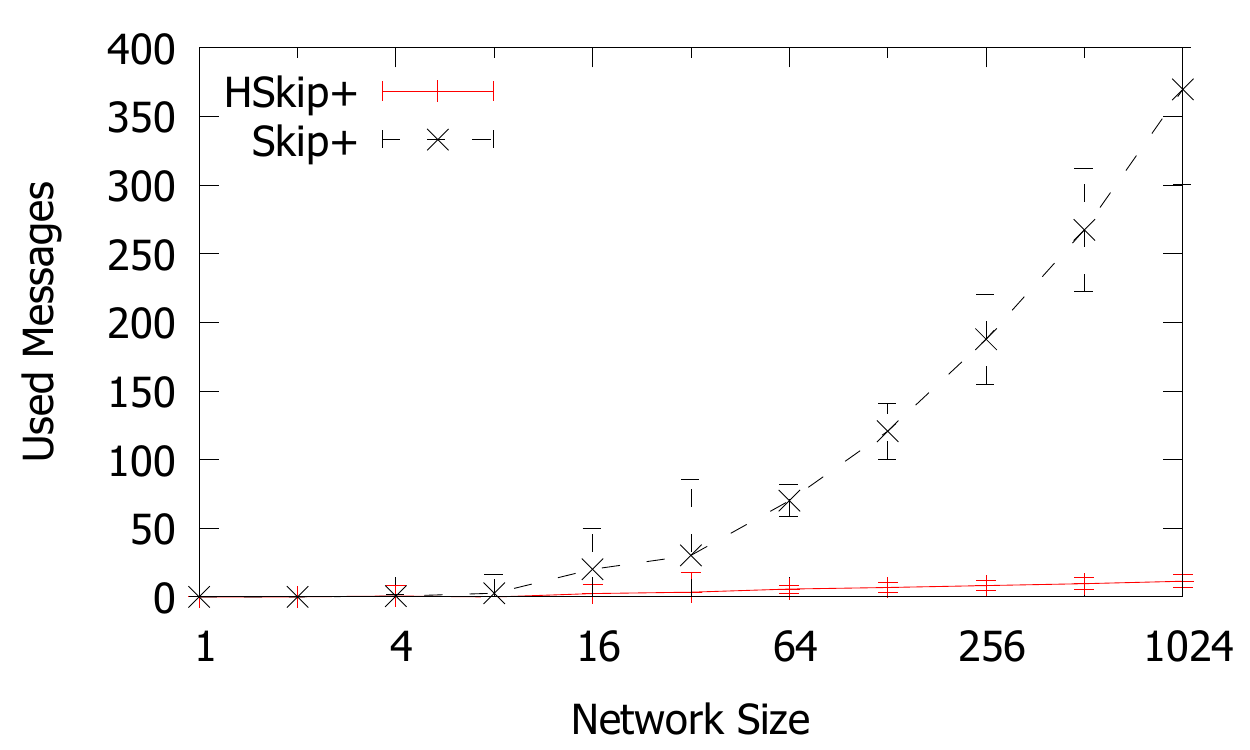}
	\vspace{-0.75cm}
	\caption[Used Messages after Leave]{Used Messages after Leave.}
	\label{fig:LeaveUsedMessages}
\end{minipage}
\vspace{-0.5cm}
\end{figure*}

\begin{figure*}[tb]
\vspace{-0.25cm}
\begin{minipage}[b]{0.48\linewidth}
\centering
\includegraphics[width=\columnwidth]{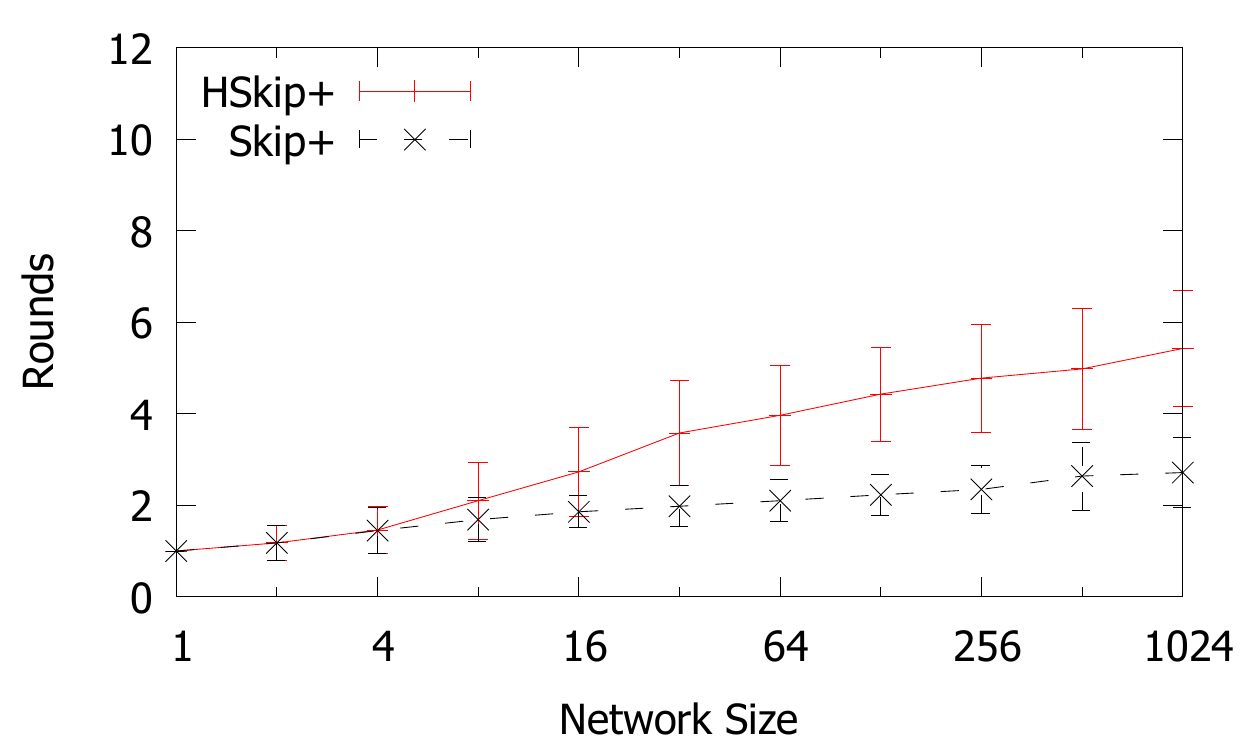}
	\vspace{-0.75cm}
	\caption[Stabilization Time after Change]{Stabilization Time after Change.}
	\label{fig:ChangeStabilization}
\end{minipage}
\hspace{0.5cm}
\begin{minipage}[b]{0.48\linewidth}
\centering
\includegraphics[width=\columnwidth]{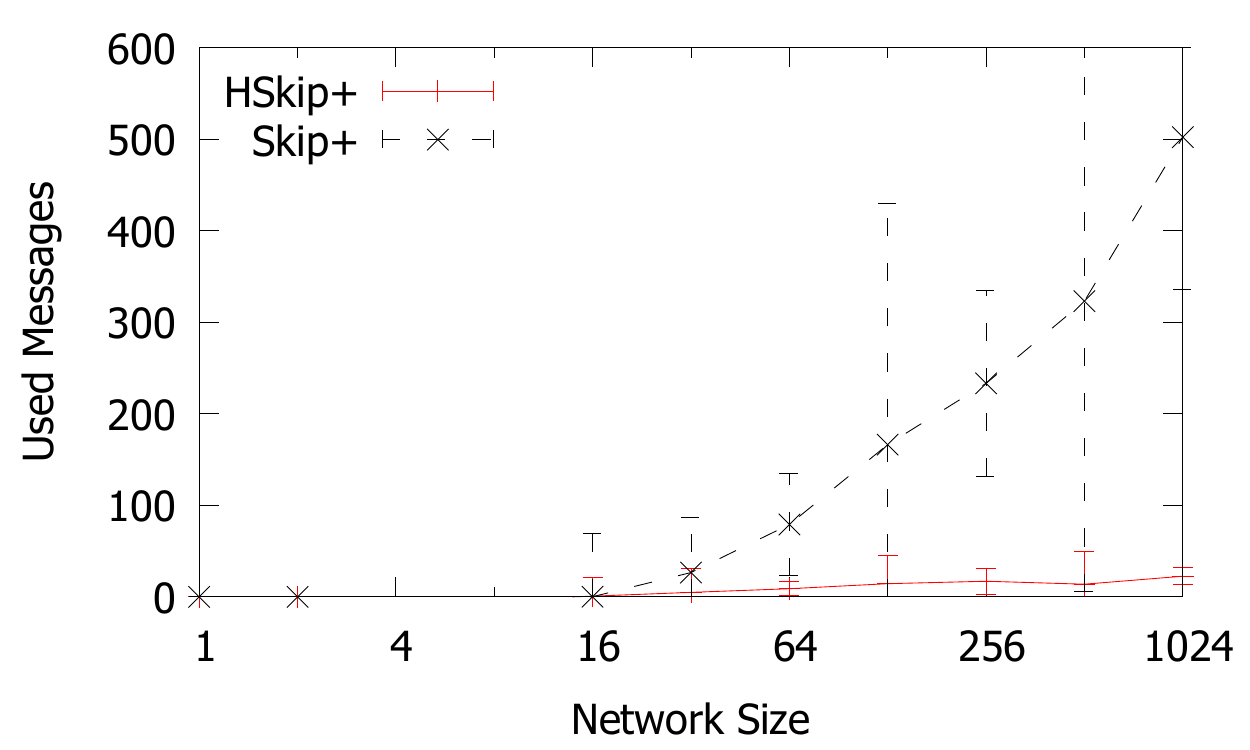}
	\vspace{-0.75cm}
	\caption[Used Messages after Change]{Used Messages after Change.}
	\label{fig:ChangeUsedMessages}
\end{minipage}
\vspace{-0.5cm}
\end{figure*}

\begin{figure*}[tb]
\vspace{-0.25cm}
\begin{minipage}[b]{0.48\linewidth}
\centering
\includegraphics[width=\columnwidth]{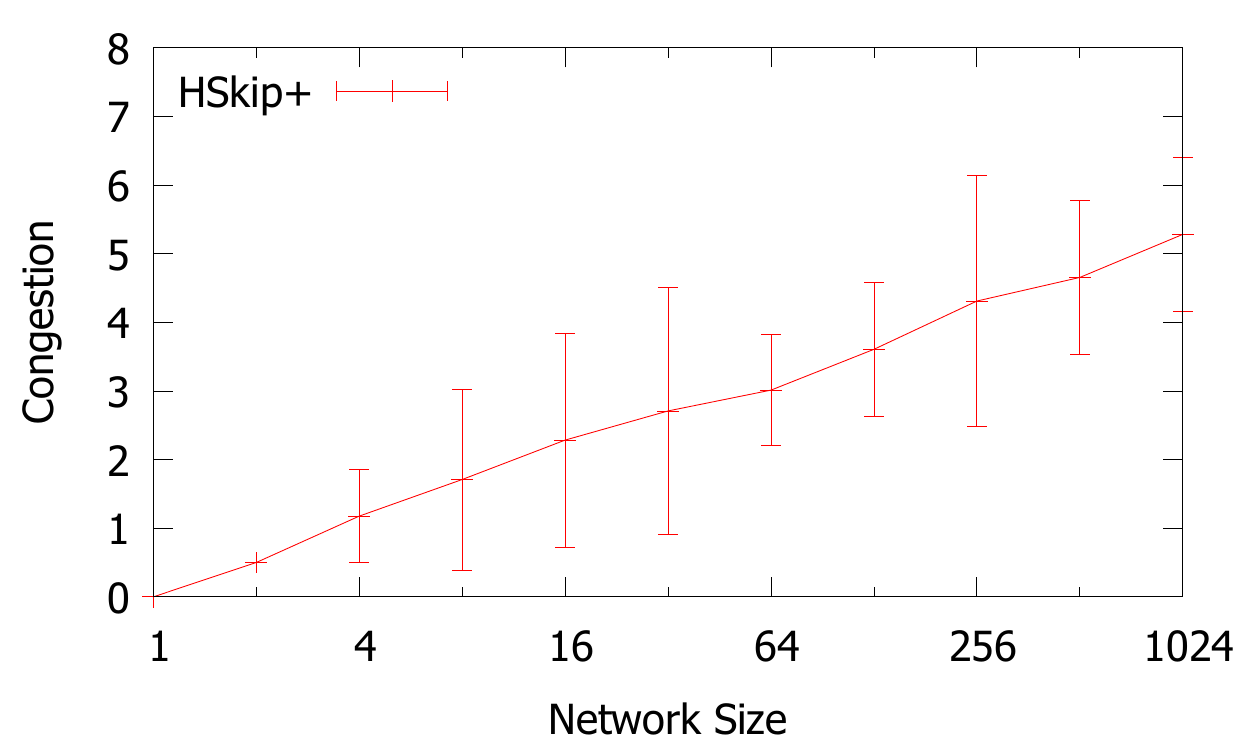}
	\vspace{-0.75cm}
	\caption[Routing Congestion of Flow Problem]{Routing Congestion of Flow Problem.}
	\label{fig:RoutingCongestion}
\end{minipage}
\hspace{0.5cm}
\begin{minipage}[b]{0.48\linewidth}
\centering
\includegraphics[width=\columnwidth]{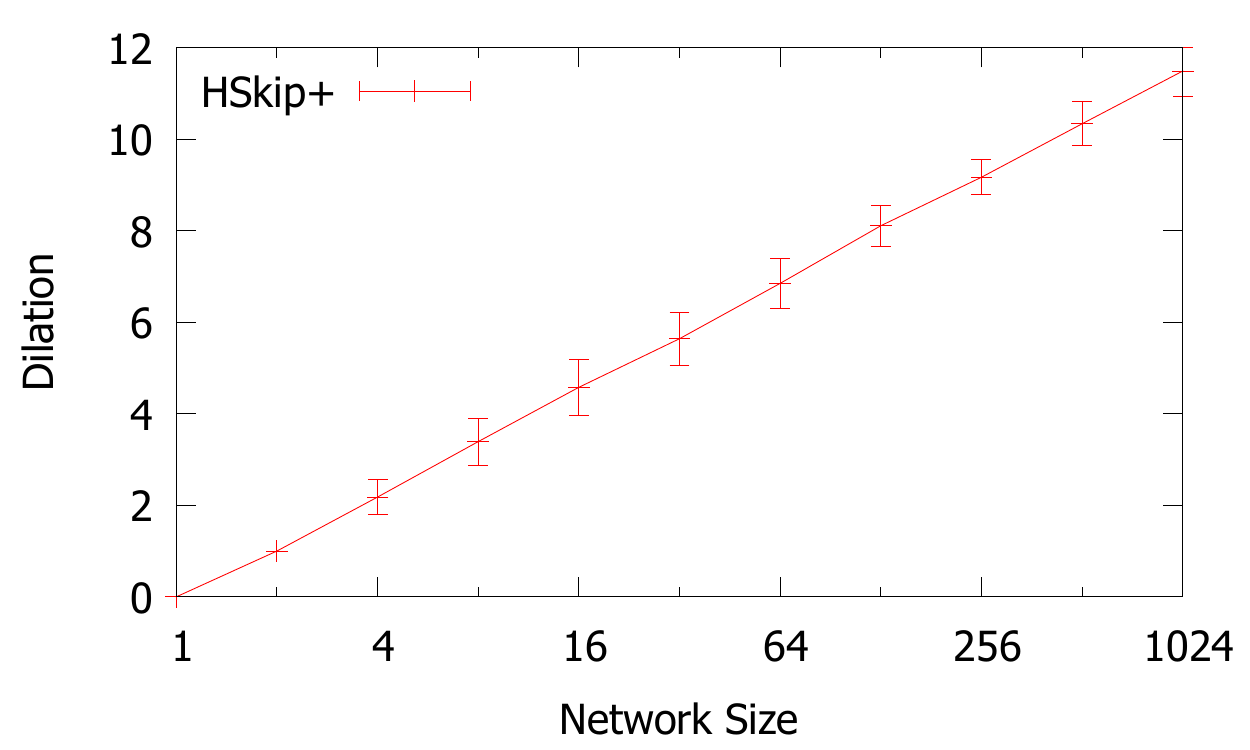}
	\vspace{-0.75cm}
	\caption[Routing Dilation of Flow Problem]{Routing Dilation of Flow Problem.}
	\label{fig:RoutingDilation}
\end{minipage}
\vspace{-0.5cm}
\end{figure*}

\begin{figure*}[tb]
\vspace{-0.25cm}
\begin{minipage}[b]{0.48\linewidth}
\centering
\includegraphics[width=\columnwidth]{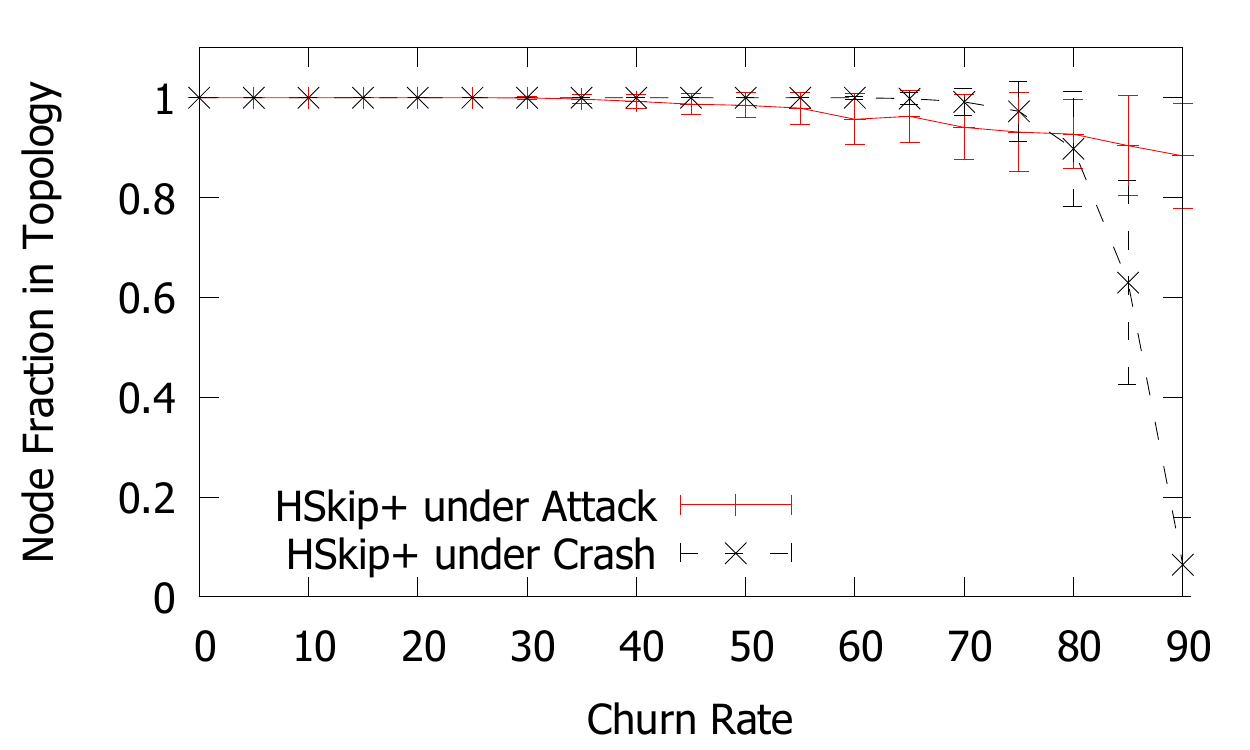}
	\vspace{-0.75cm}
	\caption[Stabilization after Crash and Attack]{Stabilization after Crash and Attack.}
	\label{fig:ChurnStabilization}
\end{minipage}
\hspace{0.5cm}
\begin{minipage}[b]{0.48\linewidth}
\centering
\includegraphics[width=\columnwidth]{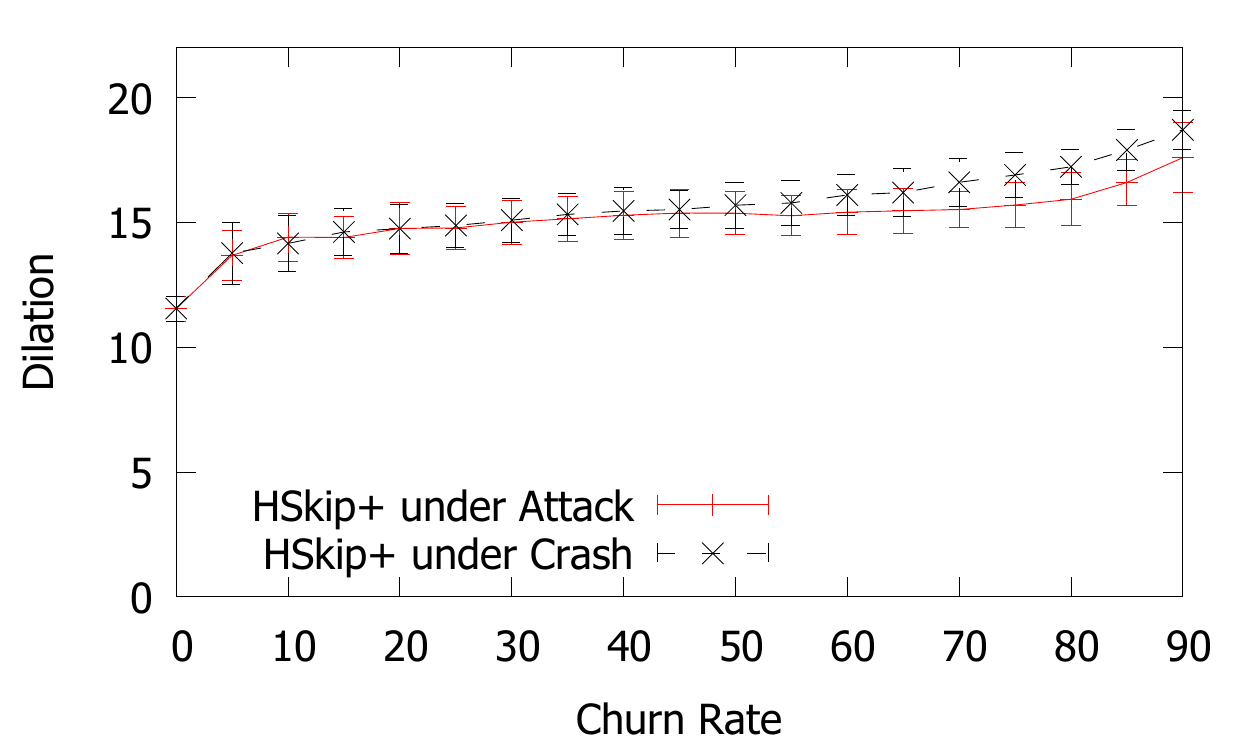}
	\vspace{-0.75cm}
	\caption[Routing Dilation during Crash and Attack]{Routing Dilation during Crash and Attack.}
	\label{fig:ChurnDilation}
\end{minipage}
\vspace{-0.5cm}
\end{figure*}

\section{Simulations}
\label{sec:Simulations}

In this section we present simulation results for the HSkip+ network in comparison to the Skip+ graph.
We simulated both protocols starting from randomly generated trees. The bandwidths were assigned randomly to the peers according to measurements of the connections in Germany \cite{BNetzA2013}.
All simulations are run 100 times with different seeds for network configurations with up to 1024 participating nodes. 
Since both networks are self-stabilizing, we focus on the stabilization costs in terms of rounds (in each round every node is allowed to process all received messages) until the desired topology is constructed from an initial weakly connected network and the number of messages which were used in this process.
The simulator can be found in the ancillary files of this paper on arXiv.org.

\subsection{Self-Stabilization}

If we look at the self-stabilization time for Skip+ and HSkip+ (cf. Fig. \ref{fig:Stabilization}), we see asymptotically similar results. Both networks need $O(\log n)$ rounds for the self-stabilization, whereas the Skip+ has slightly better results.
In contrast, if we look at the used messages (cf. Fig. \ref{fig:UsedMessages}), we see a big difference between both networks. While Skip+ uses $O(\log^4 n)$ messages for the self-stabilization process, the HSkip+ topology requires only $O(\log^2 n)$ messages.

\subsection{External Dynamics}

As the next step, we look at the external dynamics in the network, namely joining and leaving nodes, as well as bandwidth changes. 
We start with the needed rounds after a new node joined the network (cf. Fig. \ref{fig:JoinStabilization}). The result is similar to the self-stabilization we have seen before. Both networks are stabilized again after $O(\log n)$ rounds, where the Skip+ network is slightly faster.
Also the overall used messages for this process reflect the already seen result (cf. Fig. \ref{fig:JoinUsedMessages}). While the Skip+ network consumes $O(\log^4 n)$ messages, the rules of HSkip+ need only $O(\log^2 n)$ messages. The concrete number of messages are not directly comparable to the numbers in Fig.~\ref{fig:UsedMessages} as there are further messages in the system which cannot be excluded from counting.

Also the other two operations show similar results. The self-stabilization after a node left the network or changed its bandwidth can be finished in $O(\log n)$ rounds and nearly in constant time (cf. Fig. \ref{fig:LeaveStabilization} and \ref{fig:ChangeStabilization}). This tendency yields also for the used messages. The work consumed by HSkip+ is clearly less than in the case of Skip+ (cf. Fig. \ref{fig:LeaveUsedMessages} and \ref{fig:ChangeUsedMessages}).

\subsection{Routing}
\label{subsec:SimulationsRouting}

For evaluating the routing performance of the topology we look at a flow problem: For each node pair $u, v \in V$ node $u$ sends an amount of data of $\frac{u.bw \cdot v.bw}{\sum_{w\in V}{w.bw}}$ to node $v$. The average normalized congestion (according to the nodes' bandwidth) is logarithmic in the network size (cf. Fig.~\ref{fig:RoutingCongestion}). The same yields for the dilation of the routing process (cf. Fig.~\ref{fig:RoutingDilation}). Both results agree with the theoretical findings as proved in Sec.~\ref{sec:Routing}.

\subsection{Behavior under Churn}

As last aspect we look at the network behavior under churn. In a practical usage scenario of a peer-to-peer system nodes can arbitrary join and leave the network. Usual churn behavior as it was studied in several papers (i.e. \cite{DBLP:books/sp/wehrle2010/PussepLK10,DBLP:journals/ton/SteinerEB09}) has no influence on the topology as the stabilization times of our network are significantly smaller than the session and intersession times of peers in a network. Therefore, we focus on two exceptional scenarios: On the one hand we examined a crash and on the other hand an adversarial attack. In the first scenario x\% randomly chosen nodes are leaving from a network with 1024 nodes at the same time and the same number of new nodes join the network (to have a comparable size of the network). In the second scenario the leaving nodes are no longer randomly chosen, but all from a neighboring area.

The number of leaving nodes has no remarkable influence on the stabilization time as we have seen it in Fig.~\ref{fig:JoinStabilization}, Fig.~\ref{fig:LeaveStabilization} and Fig.~\ref{fig:ChangeStabilization}. Hence, we concentrate on the topology after stabilization and especially how many nodes are still connected correctly (cf. Fig.~\ref{fig:ChurnStabilization}). Up to a churn rate of about 35\% all nodes stay in the topology, under a random crash even up to a rate of 60\%. Over this limit, the topology looses nodes which are no longer reachable.
As for the second evaluation, we concentrate on the dilation. We have used the same flow problem as in Sec.~\ref{subsec:SimulationsRouting} starting at the same round as the attack and crash. In a stable network with 1024 nodes, it is about 11 hops (cf. Fig.~\ref{fig:RoutingDilation}). During the attack or crash it is just a little increased to about 14 or 15 hops; only at a very high churn rate (when we also have lost many nodes) the dilation increases noticeable.

\section{Conclusion}
\label{sec:Conclusion}

In this paper, we presented HSkip+, a self-stabilizing overlay network based on Skip+. We showed by simulations that the self-stabilization time is nearly identical in $O(\log n)$ while the improved version uses only $O(\log^2 n)$ messages for the process compared to the originally used $O(\log^4 n)$ messages. Also the dealing with external dynamics can be managed in the same time bounds and with the same work. Furthermore, we have extended the network to deal with heterogeneous bandwidths where we reach a logarithmic congestion and dilation in the routing process. Finally, the practical usage was shown by simulations under churn behavior.

\bibliographystyle{abbrv}
\bibliography{HSkip+}

\begin{thebibliography}{10}

\bibitem{Aspnes:2003:SG:644108.644170}
J.~Aspnes and G.~Shah.
\newblock Skip graphs.
\newblock In {\em Proceedings of the ACM-SIAM Symposium on Discrete Algorithms
  (SODA '03)}, pages 384--393. Society for Industrial and Applied Mathematics,
  2003.

\bibitem{Awerbuch:2004:HLD:982792.982836}
B.~Awerbuch and C.~Scheideler.
\newblock The hyperring: a low-congestion deterministic data structure for
  distributed environments.
\newblock In {\em Proceedings of the fifteenth annual ACM-SIAM symposium on
  Discrete algorithms}, SODA '04, pages 318--327, Philadelphia, PA, USA, 2004.
  Society for Industrial and Applied Mathematics.

\bibitem{Berns:2011:BSO:2050613.2050620}
A.~Berns, S.~Ghosh, and S.~V. Pemmaraju.
\newblock Building self-stabilizing overlay networks with the transitive
  closure framework.
\newblock In {\em Proceedings of the International Conference on Stabilization,
  Safety, and Security of Distributed Systems (SSS '11)}, pages 62--76, 2011.

\bibitem{Bhargava:2004:PDO:1007912.1007938}
A.~Bhargava, K.~Kothapalli, C.~Riley, C.~Scheideler, and M.~Thober.
\newblock Pagoda: a dynamic overlay network for routing, data management, and
  multicasting.
\newblock In {\em Proceedings of the ACM Symposium on Parallelism in Algorithms
  and Architectures (SPAA '04)}, pages 170--179, 2004.

\bibitem{BNetzA2013}
{Bundesnetzagentur f\"ur Elektrizit\"at, Gas, Telekommunikation, Post und
  Eisenbahnen}.
\newblock T\"atigkeitsbericht 2012/2013, 2013.

\bibitem{Dijkstra:1974:SSS:361179.361202}
E.~W. Dijkstra.
\newblock Self-stabilizing systems in spite of distributed control.
\newblock {\em Communications of the ACM}, 17(11):643--644, Nov. 1974.

\bibitem{Dolev:2004:HSP:1025126.1025933}
S.~Dolev and R.~I. Kat.
\newblock {HyperTree for Self-Stabilizing Peer-to-Peer Systems}.
\newblock In {\em Proceedings of the IEEE International Symposium on Network
  Computing and Applications (NCA '04)}, pages 25--32, 2004.

\bibitem{Dolev:2013:SDS:2451995.2452270}
S.~Dolev and N.~Tzachar.
\newblock Spanders: Distributed spanning expanders.
\newblock {\em Sci. Comput. Program.}, 78(5):544--555, May 2013.

\bibitem{Gall:2010:TCD:2128719.2128746}
D.~Gall, R.~Jacob, A.~Richa, C.~Scheideler, S.~Schmid, and H.~T\"{a}ubig.
\newblock Time complexity of distributed topological self-stabilization: the
  case of graph linearization.
\newblock In {\em Proceedings of the Latin American Conference on Theoretical
  Informatics (LATIN '10)}, pages 294--305, 2010.

\bibitem{Jacob:2009:DPT:1582716.1582741}
R.~Jacob, A.~Richa, C.~Scheideler, S.~Schmid, and H.~T\"{a}ubig.
\newblock A distributed polylogarithmic time algorithm for self-stabilizing
  skip graphs.
\newblock In {\em Proceedings of the ACM Symposium on Principles of Distributed
  Computing (PODC '09)}, pages 131--140, 2009.

\bibitem{Jacob:2012:THT:2364637.2364951}
R.~Jacob, S.~Ritscher, C.~Scheideler, and S.~Schmid.
\newblock {Towards higher-dimensional topological self-stabilization: A
  distributed algorithm for Delaunay graphs}.
\newblock {\em Theor. Comput. Sci.}, 457:137--148, Oct. 2012.

\bibitem{Kniesburges:2011:RSC:1989493.1989527}
S.~Kniesburges, A.~Koutsopoulos, and C.~Scheideler.
\newblock {Re-Chord: a self-stabilizing chord overlay network}.
\newblock In {\em Proceedings of the ACM Symposium on Parallelism in Algorithms
  and Architectures (SPAA '11)}, pages 235--244, 2011.

\bibitem{Kniesburges2013}
S.~Kniesburges, A.~Koutsopoulos, and C.~Scheideler.
\newblock {CONE-DHT: A Distributed self-stabilizing algorithm for a
  heterogeneous storage system}.
\newblock In {\em Proceedings of the International Symposium on Distributed
  Computing (DISC'13)}, 2013.

\bibitem{MW12}
R.~Meier and R.~Wattenhofer.
\newblock {Peer-to-Peer Streaming in Heterogeneous Environments}.
\newblock {\em Signal Processing: Image Communication}, 27(5):457--469, March
  2012.

\bibitem{Nejdl:2003:SRC:775152.775229}
W.~Nejdl, M.~Wolpers, W.~Siberski, C.~Schmitz, M.~Schlosser, I.~Brunkhorst, and
  A.~L\"{o}ser.
\newblock {Super-peer-based routing and clustering strategies for RDF-based
  peer-to-peer networks}.
\newblock In {\em Proceedings of the ACM International Conference on World Wide
  Web (WWW '03)}, pages 536--543, 2003.

\bibitem{Nor:2011:CSD:2050613.2050640}
R.~M. Nor, M.~Nesterenko, and C.~Scheideler.
\newblock Corona: a stabilizing deterministic message-passing skip list.
\newblock In {\em Proceedings of the International Conference on Stabilization,
  Safety, and Security of Distributed Systems (SSS '11)}, pages 356--370, 2011.

\bibitem{DBLP:books/sp/wehrle2010/PussepLK10}
K.~Pussep, C.~Leng, and S.~Kaune.
\newblock {Modeling User Behavior in P2P Systems}.
\newblock In K.~Wehrle, M.~G{\"u}nes, and J.~Gross, editors, {\em Modeling and
  Tools for Network Simulation}, pages 447--461. Springer, 2010.

\bibitem{Scheideler:2009:DOH:1575973.1576024}
C.~Scheideler and S.~Schmid.
\newblock {A Distributed and Oblivious Heap}.
\newblock In {\em Proceedings of the Internatilonal Collogquium on Automata,
  Languages and Programming (ICALP '09)}, pages 571--582, 2009.

\bibitem{Schneider:1993:SEL:151254.151256}
M.~Schneider.
\newblock Self-stabilization.
\newblock {\em ACM Computing Surveys}, 25(1):45--67, Mar. 1993.

\bibitem{Shaker:2005:SSR:1099548.1100573}
A.~Shaker and D.~S. Reeves.
\newblock {Self-Stabilizing Structured Ring Topology {P2P} Systems}.
\newblock In {\em Proceedings of the IEEE International Conference on
  Peer-to-Peer Computing (P2P '05)}, pages 39--46, 2005.

\bibitem{Srivatsa:2004:SUP:1009385.1010039}
M.~Srivatsa, B.~Gedik, and L.~Liu.
\newblock {Scaling Unstructured Peer-to-Peer Networks With Multi-Tier
  Capacity-Aware Overlay Topologies}.
\newblock In {\em Proceedings of the IEEE Parallel and Distributed Systems,
  Tenth International Conference (ICPADS '04)}, 2004.

\bibitem{DBLP:journals/ton/SteinerEB09}
M.~Steiner, T.~En-Najjary, and E.~W. Biersack.
\newblock {Long term study of peer behavior in the KAD DHT}.
\newblock {\em IEEE/ACM Trans. Netw.}, 17(5):1371--1384, 2009.

\end{thebibliography}

\end{document}